\documentclass[twoside,leqno,twocolumn]{article}
\usepackage[hyperfootnotes = false]{hyperref}
\usepackage{ltexpprt}

\usepackage{amsfonts}
\usepackage[cmex10]{amsmath}
\usepackage{hyperref}
\usepackage[]{notes}
\usepackage[ruled, vlined, nofillcomment]{algorithm2e}
\usepackage{subcaption}
\usepackage{cite}
\usepackage{xspace}
\usepackage{url}
\usepackage{icomma}
\usepackage{booktabs}
\usepackage[english]{babel}
\usepackage{flushend}
\usepackage{tikz}
\usepackage{pgfplots}
\usetikzlibrary{fit,positioning,shapes}
\usepackage{wrapfig}

\usepackage{thmtools}
\usepackage{thm-restate}

\newcommand{\squishlist}{
\begin{list}{$\bullet$}
 {  \setlength{\itemsep}{0pt}
    \setlength{\parsep}{3pt}
     \setlength{\topsep}{3pt}
     \setlength{\partopsep}{0pt}
     \setlength{\leftmargin}{2em}
     \setlength{\labelwidth}{1.5em}
     \setlength{\labelsep}{0.5em}
} }
\newcommand{\squishlisttight}{
 \begin{list}{$\bullet$}
  { \setlength{\itemsep}{0pt}
    \setlength{\parsep}{0pt}
    \setlength{\topsep}{0pt}
    \setlength{\partopsep}{0pt}
    \setlength{\leftmargin}{2em}
    \setlength{\labelwidth}{1.5em}
    \setlength{\labelsep}{0.5em}
} }
\newcommand{\squishdesc}{
 \begin{list}{}
  {  \setlength{\itemsep}{0pt}
     \setlength{\parsep}{3pt}
     \setlength{\topsep}{3pt}
     \setlength{\partopsep}{0pt}
     \setlength{\leftmargin}{1em}
     \setlength{\labelwidth}{1.5em}
     \setlength{\labelsep}{0.5em}
} }
\newcommand{\squishend}{
  \end{list}
}

\newcommand{\comment}[1]{}

\newcommand{\qed}{\hfill $\Box$}

\newtheorem{problem}{Problem}
\newtheorem{observation}{Observation}

\newcommand{\spara}[1]{\smallskip\noindent{\bf{#1}}}

\newcommand{\Reals}{\ensuremath{\mathbb{R}}}

\newcommand{\NP}{\ensuremath{\mathbf{NP}}}
\newcommand{\bigO}{\ensuremath{{\cal O}}}

\newcommand{\Property}{\ensuremath{\rho}} 
\newcommand{\indicator}{\ensuremath{\mathbb{I}}}

\newcommand{\stars}{\ensuremath{s}}
\newcommand{\density}{\ensuremath{d}}
\newcommand{\connectivity}{\ensuremath{c}}
\newcommand{\abs}[1]{{\left|#1\right|}}

\newcommand{\set}[1]{\left\{#1\right\}}

\newcommand{\funcdef}[3]{{#1}:{#2} \to {#3}}

\newcommand{\sets}{\ensuremath{\mathcal{S}}}
\newcommand{\communities}{\ensuremath{\mathcal{C}}}
\newcommand{\community}{\ensuremath{C}}

\newcommand{\avg}{\ensuremath{\mathrm{avg}}}
\newcommand{\edges}{\ensuremath{E}}
\newcommand{\indedges}{\ensuremath{E_0}}
\newcommand{\indgraph}{\ensuremath{G_0}}
\newcommand{\sparseG}{\ensuremath{G^{*}}}
\newcommand{\sparseedges}{\ensuremath{E^{*}}}
\newcommand{\ratio}{\ensuremath{\rho}}
\newcommand{\reldegree}{\ensuremath{\delta}}
\newcommand{\relSPlen}{\ensuremath{\lambda}}

\newcommand{\intersection}{\ensuremath{\mathsf{I}}}
\newcommand{\union}{\ensuremath{\mathsf{U}}}

\newcommand{\densitypotential}{\ensuremath{\Phi}}
\newcommand{\condensitypotential}{\ensuremath{\Psi}}
\newcommand{\conpotential}{\ensuremath{U}}
\newcommand{\cc}{\ensuremath{\mathit{cc}}}

\newcommand{\setnetworkprb}{\textsc{Net\-Sparse}\xspace}
\newcommand{\weightedproblem}{\textsc{Weighted\-Net\-Sparse}\xspace}
\newcommand{\connectedprb}{\textsc{Sparse\-Conn}\xspace}
\newcommand{\densprb}{\textsc{Sparse\-Dens}\xspace}

\newcommand{\hitprb}{\textsc{Hitting\-Set}\xspace}
\newcommand{\starprb}{\textsc{Sparse\-Stars}\xspace}
\newcommand{\distarprb}{\textsc{Sparse\-Di\-Stars}\xspace}
\newcommand{\matchprb}{\textsc{3D-Matching}\xspace}
\newcommand{\hyperprb}{\textsc{Hyper\-edge\-Matching}\xspace}

\newcommand{\hypergreedy}{\texttt{HGreedy}}
\newcommand{\stargreedy}{\texttt{SGreedy}}
\newcommand{\starrandom}{\texttt{SRandom}}
\newcommand{\matching}{\texttt{DS2S}}

\newcommand{\densitygreedy}{\texttt{DGreedy}}
\newcommand{\densityrandom}{\texttt{DRandom}}
\newcommand{\densitysort}{\texttt{DSort}}
\newcommand{\LS}{\texttt{LS}}

\newcommand{\amazon}{{\textsl{Amazon}}\xspace}
\newcommand{\DBLP}{{\textsl{DBLP}}\xspace}
\newcommand{\youtube}{{\textsl{Youtube}}\xspace}
\newcommand{\cocktails}{{\textsl{Cocktails}}\xspace}
\newcommand{\birds}{{\textsl{Birds}}\xspace}
\newcommand{\KDD}{{\textsl{KDD}}\xspace}
\newcommand{\ICDM}{{\textsl{ICDM}}\xspace}
\newcommand{\FBcirc}{{\textsl{FB-circles}}\xspace}
\newcommand{\FBfeat}{{\textsl{FB-features}}\xspace}
\newcommand{\lastFMart}{{\textsl{lastFM-artists}}\xspace}
\newcommand{\lastFMtag}{{\textsl{lastFM-tags}}\xspace}
\newcommand{\DBbook}{{\textsl{DB-bookmarks}}\xspace}
\newcommand{\DBtag}{{\textsl{DB-tags}}\xspace}
\newcommand{\dataset}[1]{{\textsl{D#1}}\xspace}

\newcommand{\etal}{{et al.}}

\pgfdeclarelayer{background}
\pgfdeclarelayer{foreground}
\pgfsetlayers{background,main,foreground}

\makeatletter
\tikzset{multicircle/.style  args={#1, #2}{%
 alias=tmp@name, %
  postaction={%
    insert path={
     \pgfextra{%
     \pgfpointdiff{\pgfpointanchor{\pgf@node@name}{center}}%
                  {\pgfpointanchor{\pgf@node@name}{east}}%
     \pgfmathsetmacro\insiderad{\pgf@x}%
        \fill[white] (\pgf@node@name.center)  circle (\insiderad-\pgflinewidth);%
        \draw[#2] (\pgf@node@name.center)  circle (\insiderad-\pgflinewidth);%
        \fill[#2] (\pgf@node@name.center)  -- ++(0:\insiderad-\pgflinewidth) arc (0:#1:\insiderad-\pgflinewidth)--cycle;%
        }}}}}
\makeatother

\definecolor{yafaxiscolor}{rgb}{0.3, 0.3, 0.3}
\definecolor{yafcolor1}{rgb}{0.4, 0.165, 0.553}
\definecolor{yafcolor2}{rgb}{0.949, 0.482, 0.216}
\definecolor{yafcolor3}{rgb}{0.47, 0.549, 0.306}
\definecolor{yafcolor4}{rgb}{0.925, 0.165, 0.224}
\definecolor{yafcolor5}{rgb}{0.141, 0.345, 0.643}
\definecolor{yafcolor6}{rgb}{0.965, 0.933, 0.267}
\definecolor{yafcolor7}{rgb}{0.627, 0.118, 0.165}
\definecolor{yafcolor8}{rgb}{0.878, 0.475, 0.686}
\definecolor{yafcolor9}{rgb}{0.965, 0.733, 0.767}

\newlength{\yafaxispad}
\newlength{\yaftlpad}
\newlength{\yaflabelpad}
\newlength{\yafaxiswidth}
\newlength{\yafticklen}
\makeatletter
\def\pgfplots@drawtickgridlines@INSTALLCLIP@onorientedsurf#1{}
\makeatother

\newcommand{\yafdrawxaxis}[2]{
    \pgfplotstransformcoordinatex{#1}\let\xmincoord=\pgfmathresult 
    \pgfplotstransformcoordinatex{#2}\let\xmaxcoord=\pgfmathresult 
    \pgfsetlinewidth{\yafaxiswidth} 
    \pgfsetcolor{yafaxiscolor}
    \pgfpathmoveto{\pgfpointadd{\pgfpointadd{\pgfplotspointrelaxisxy{0}{0}}{\pgfqpointxy{\xmincoord}{0}}}{\pgfqpoint{-0.5\yafaxiswidth}{\yafaxispad}}}
    \pgfpathlineto{\pgfpointadd{\pgfpointadd{\pgfplotspointrelaxisxy{0}{0}}{\pgfqpointxy{\xmaxcoord}{0}}}{\pgfqpoint{0.5\yafaxiswidth}{\yafaxispad}}}
    \pgfusepath{stroke}

}
\newcommand{\yafdrawyaxis}[2]{
    \pgfplotstransformcoordinatey{#1}\let\ymincoord=\pgfmathresult 
    \pgfplotstransformcoordinatey{#2}\let\ymaxcoord=\pgfmathresult 
    \pgfsetlinewidth{\yafaxiswidth} 
    \pgfsetcolor{yafaxiscolor}
    \pgfpathmoveto{\pgfpointadd{\pgfpointadd{\pgfplotspointrelaxisxy{0}{0}}{\pgfqpointxy{0}{\ymincoord}}}{\pgfqpoint{\yafaxispad}{-0.5\yafaxiswidth}}}
    \pgfpathlineto{\pgfpointadd{\pgfpointadd{\pgfplotspointrelaxisxy{0}{0}}{\pgfqpointxy{0}{\ymaxcoord}}}{\pgfqpoint{\yafaxispad}{0.5\yafaxiswidth}}}
    \pgfusepath{stroke}
}

\newcommand{\yafdrawaxis}[4]{\yafdrawxaxis{#1}{#2}\yafdrawyaxis{#3}{#4}}

\pgfplotscreateplotcyclelist{yaf}{%
{yafcolor1,mark options={scale=0.75},mark=o}, 
{yafcolor2,mark options={scale=0.75},mark=square},
{yafcolor3,mark options={scale=0.75},mark=triangle},
{yafcolor4,mark options={scale=0.75},mark=o},
{yafcolor5,mark options={scale=0.75},mark=o},
{yafcolor6,mark options={scale=0.75},mark=o},
{yafcolor7,mark options={scale=0.75},mark=o},
{yafcolor8,mark options={scale=0.75},mark=o}} 

\pgfplotsset{axis y line=left, axis x line=bottom,
    tick align=outside,
    compat = 1.3,
    tickwidth=\yafticklen,
    clip = false,
    every axis title shift = 0pt,
    x axis line style= {-, line width = 0pt, opacity = 0},
    y axis line style= {-, line width = 0pt, opacity = 0},
    x tick style= {line width = \yafaxiswidth, color=yafaxiscolor, yshift = \yafaxispad},
    y tick style= {line width = \yafaxiswidth, color=yafaxiscolor, xshift = \yafaxispad},
    x tick label style = {font=\scriptsize, yshift = \yaftlpad},
    y tick label style = {font=\scriptsize, xshift = \yaftlpad},
    every axis y label/.style = {at = {(ticklabel cs:0.5)}, rotate=90, anchor=center, font=\scriptsize, yshift = -\yaflabelpad},
    every axis x label/.style = {at = {(ticklabel cs:0.5)}, anchor=center, font=\scriptsize, yshift = \yaflabelpad},
    x tick label style = {font=\scriptsize, yshift = 1pt},
    grid = major,
    major grid style  = {dash pattern = on 1pt off 3 pt},
    every axis plot post/.append style= {line width=\yafaxiswidth} ,
    legend cell align = left,
    legend style = {inner ysep = 1pt, inner xsep = 3pt, cells = {font=\scriptsize}},
    legend image code/.code={%
        \draw[mark repeat=2,mark phase=2,#1] 
        plot coordinates { (0cm,0cm) (0.2cm,0cm) (0.4cm,0cm) };%
    } 
}

\begin{document}
\clubpenalty=10000
\widowpenalty = 10000

\title{Community-aware network sparsification}

\author{
Aristides Gionis\quad Polina Rozenshtein\quad Nikolaj Tatti  \\
Aalto University \\ Espoo, Finland \\
\texttt{firstname.lastname@aalto.fi} \\
\and
Evimaria Terzi \\ 
Boston University \\ Boston, USA \\
\texttt{evimaria@bu.edu} \\
}
\date{}

\maketitle




\begin{abstract}\small\baselineskip=9pt
Network sparsification aims to reduce the number of edges of a network while maintaining its
structural properties:  
shortest paths, cuts, spectral measures, or network modularity. 
Sparsification has multiple applications, such as, 
speeding up graph-mining algorithms, graph visualization, 
as well as identifying the important network edges.

In this paper, we consider a novel formulation 
of the network-sparsification problem. 
In addition to the network, 
we also consider as input a set of communities. 
The goal is to sparsify the network so as to preserve the network structure
with respect to the given communities.
%
We introduce 
two variants
of the community-aware sparsification problem, 
leading to sparsifiers that satisfy different {\em connectedness} community properties. 
From the technical point of view, we prove hardness results and devise 
effective approximation algorithms. 
Our experimental results on a large collection of datasets demonstrate the effectiveness of our algorithms.

\end{abstract}

\section{Introduction}
\label{sec:intro}

Large graphs, or networks, arise in many applications, 
e.g., social networks, information networks, and biological networks.
Real-world networks are usually sparse,
meaning that the actual number of edges in the network $m$
is much smaller than $\bigO(n^2)$, 
where $n$ is the number of network nodes. 
Nonetheless, in practice, 
it is common to work with networks whose average degree is in the order of hundreds or thousands, 
leading to many computational and data-analysis challenges.
 
Sparsification is a fundamental operation
that aims to reduce the number of edges of a network
while maintaining its structural properties.
Sparsification has numerous applications, such as, 
graph summarization and visualization,
speeding up graph algorithms,  and
identification of important edges.
%
A number of different sparsification methods have been proposed, 
depending on the network property that one aims to preserve. 
Typical properties include 
paths and connectivity~\cite{elkin2005approximating,zhou10network},
cuts~\cite{ahn12graph,fung11general},
and
spectral properties~\cite{ahn13spectral,batson13spectral,spielman11graph}.

Existing work on network sparsification ignores the fact that 
the observed network is the result of different latent factors. 
For instance, imagine a Facebook user who posts a high-school photo, 
which leads to a discussion thread among old high-school friends.
In this case, the participation of users in a discussion group
is a result of an underlying community. 
In general, the network structure reflects 
a number of underlying (and potentially overlapping) communities.
Thus, if it is this community structure that guides the network-formation process, 
then the community structure should
also be taken into consideration in network sparsification.

Motivated by this view, 
we consider the following problem:
{\em 
Given a network and a set of potentially over\-lap\-ping communities,
sparsify the network so as to preserve its structural properties
with respect to the given communities.}
Our goal is to find a small set of edges
that best summarize, or explain, a given commu\-nity structure in the network.

Our setting has many interesting applications. 
For example, consider a social network where users discuss various topics. 
Each topic defines a community of people  
interested in the topic. 
Given a set of topics, 
we want to find a sparse network that best explains the respective communities. 
Similar problems arise in collaboration networks, 
where communities are defined by collaboration themes, 
consumer networks where communities are defined by products, etc.
Finding an optimal sparse network with respect to a set of communities
is a means of understanding the interplay between network structure and 
content-induced communities.  

We formalize the above intuition by defining the  {\setnetworkprb} problem:
given an undirected graph $G=(V,E)$
and a set of communities $\communities = \{\community_1,\ldots ,\community_\ell\}$ over $V$, 
we ask to find a sparsified graph $G'=(V',E')$
with $V'=\cup_{i=1,\ldots ,\ell}\community_i$ and $E'\subseteq E$, 
so as to minimize  
$|E'|$ and guarantee that
every graph $G'[\community_i]$, induced by the nodes in the community $\community_i$,
satisfies a certain \emph{connectedness requirement}.  

Different connectedness requirements give rise to different variants of the 
{\setnetworkprb} problem.
We consider three such requirements: 
$(i)$ \emph{connectivity}, $(ii)$ \emph{density} and $(iii)$ \emph{star containment}.
While connectivity has been addressed by previous work~\cite{angluin13connectivity},
we are the first to introduce and study the other two properties,
which define the {\densprb} and {\starprb} problems, respectively.
In the {\densprb} problem the requirement is that 
each induced graph $G'[\community_i]$ has a minimum density requirement. 
In the {\starprb} problem the requirement is that $G'[\community_i]$ 
contains a star as a subgraph. 
We establish the computa\-tional complexity of the two problems, {\densprb} and {\starprb},  
and present approximation algorithms for solving them.



An interesting special case arises when
the input to our problem consists only of the collection of communities
and there is no network $G=(V,E)$.
In this case, we can consider that $G$ is the complete graph (clique) 
and the {\setnetworkprb} becomes a \emph{network design} problem, 
where the goal is to construct a network that satisfies the connectedness requirement among
the nodes in the different communities.

The list of our contributions is the following.
\squishlist
\item We introduce the novel problem
of sparsifying a network while
preserving the structure of a given set of communities.
\item 
We formulate different variants of this
{\em network-aware sparsification} task, 
by considering preserving connectedness properties within communities.
\item 
For the proposed formulations
we present complexity results 
and efficient approximation algorithms.
\item 
We present experimental results on a large collection of real datasets, 
demonstrating that our algorithms effectively sparsify the
underlying network while maintaining the required community structure and 
other key properties of the original graph.
\squishend

We note that our implementation and datasets will be publicly available. 
Proofs, other results, and additional experiments are in the supplementary material.

\section{General problem definition}
\label{sec:prel}


Our input consists of an underlying undirected graph $G=(V, E)$ 
having $|V|=n$ vertices and $|E|=m$ edges.
As a special case, the underlying network $G$ can be {\em empty}, 
i.e., there is no underlying network at all. 
We treat this case equivalently to the case in which the underlying network 
is the complete graph (clique).

Additionally, 
we consider as input
a collection of~$\ell$ sets 
$\communities = \{\community_1,\ldots ,\community_\ell\}$ over $V$, 
i.e., $\community_i\subseteq V$.
We think of the sets {\communities} and we refer to them as network {\em communities}.
We assume that the sets in $\communities$ may be overlapping.

Our objective is to find a \emph{sparsifier} of the network $G$
that maintains certain {\em connectedness} properties
with respect to the given communities $\communities$.
A sparsifier of $G$
is a subgraph $G'=(V',E')$, 
where the number of edges $|E'|$ is significantly smaller than $|E|$.
The vertices $V'$ spanned by $G'$ 
are the vertices that appear in at least one community $\community_i$, i.e., 
$V'=\cup_{i=1}^\ell \community_i$.
Without loss of generality
we assume that $\cup_{i=1}^\ell \community_i = V$, so $V'=V$. 

\spara{Connectedness properties:}
To formally define the sparsification problem, 
we need to specify what it means for the sparse network
to satisfy a connectedness property with respect to the set of communities \communities. 

We provide the following formalization: 
given a graph $G=(V, E)$ and $S\subseteq V$, 
we use $E(S)$ to denote the edges of $E$ that have both endpoints in~$S$, 
and $G(S)=(S, E(S))$ is the subgraph of $G$ \emph{induced} by~$S$.
We are interested in whether a graph $G=(V,E)$ 
satisfies a certain property $\Property$ for a given set of communities
$\communities = \{\community_1,\ldots ,\community_\ell\}$ where $\community_i\subseteq V$.
We say that $G$ satisfies property $\Property$ with respect to a community $\community_i$ 
if the induced subgraph $G(\community_i)$ satisfies property $\Property$. 
We write $\indicator_\Property(G, \community_i)=1$ 
to denote the fact that $G(\community_i)$ satisfies property $\Property$, 
and $\indicator_\Property(G, \community_i)=0$ otherwise.

We consider three graph properties: 
($i$) \emph{connectivity}, denoted by $\connectivity$; 
($ii$) \emph{density}, denoted by $\density$; and 
($iii$) \emph{star containment}, denoted by $\stars$.
The corresponding indicator functions are denoted by 
$\indicator_\connectivity$, 
$\indicator_{\density\ge\alpha_i}$, 
and 
$\indicator_{\stars}$.

The connectivity property requires that each set $\community_i$
induces a connected subgraph.
The density property requires that each set $\community_i$ 
induces a subgraph of density at least~$\alpha_i$.
The density property is motivated by the common perception that
communities are usually densely connected.
The star-containment property requires that each set $\community_i$ induces a graph that contains a star.
The intuition is that star-shaped communities have small diameter and also have a community ``leader,'' which corresponds to the center of the graph.

\spara{Problem definition:}
We can now define the general problem that we study in this paper.

\begin{problem}[{\setnetworkprb}]
\label{problem:general}
Consider a network $G=(V, E)$, 
and let $\Property$ be a graph property.
Given a set of $\ell$ communities
$\communities = \{\community_1,\ldots, \community_\ell\}$, 
we want to find a {\em sparse network} $G' = (V, E')$ so that
{\em ($i$)} $E' \subseteq E$; 
{\em ($ii$)} $G'$ satisfies property $\Property$ for all communities $\community_i\in\communities$; 
and 
{\em ($iii$)} the total number of edges (or total edge weight, if defined) on the sparse network $|E'|$
is minimized. 
\end{problem}


One question is whether a feasible solution for problem {\setnetworkprb} exists.
This can be easily checked by testing if property $\Property$
is satisfied for each $\community_i$ in the original network $G$.
If this is true, then a feasible solution exists --- the original network $G$ is such a solution.
Furthermore, if property $\Property$ is not satisfied for a community $\community_i$ in the original network, 
then this community can be dropped, 
and a feasible solution exists for all the communities for which the property is satisfied in the original network.

One should also note that the problem complexity and the
algorithms for solving Problem~{\ref{problem:general}} 
depend on the property $\Property$. 
This is illustrated in the next paragraph, as well as in the next two sections.

\spara{Connectivity.}
Angluin {\etal}~\cite{angluin13connectivity} study the {\setnetworkprb} problem for the connectivity property.
They show that it is an \NP-hard problem 
and provide an algorithm with logarithmic approximation guarantee.


\section{Sparsification with density constraints}
\label{sec:density}

We assume that each community $\community_i\in \communities$ 
is associated with a \emph{density requirement}
$\alpha_i$, where $0\leq \alpha_i\leq 1$.
This is the target density for community $\community_i$ in the sparse network. 
As a special case all communities may have the same target density,
i.e., $\alpha_i = \alpha$.
We say that a network $G'=(V,E')$ 
satisfies the density property
with respect to a community $\community_i$ and density threshold $\alpha_i$, 
if $|E'(\community_i)|\geq \alpha_i {|\community_i|\choose 2}$, 
that is, the density of the subgraph induced by $\community_i$ in $G'$ is at least $\alpha_i$. 
We denote this by 
$\indicator_{\density\ge\alpha_i}(G',\community_i)=1$;
otherwise we set 
$\indicator_{\density\ge\alpha_i}(G',\community_i)=0$.

The {\densprb} problem is defined as the special case of Problem~{\ref{problem:general}}, 
where $\Property$ is the density property. 
Before presenting our algorithm for the {\densprb} problem, 
we first establish its complexity. 

\begin{proposition}
\label{proposition:density-nphardness}
The \densprb\ problem is \NP-hard.
\end{proposition}

We now present {\densitygreedy}, 
a greedy algorithm for the {\densprb} problem.
Given an instance of {\densprb}, 
i.e., a network $G=(V,E)$, 
a set of $\ell$ communities $\community_i$, 
and corresponding densities $\alpha_i$, 
the algorithm provides an $\bigO(\log \ell)$-approximation guarantee.

To analyze {\densitygreedy}, 
we consider 
a \emph{potential function}~\densitypotential,
defined over subsets of edges of $E$.
For an edge set $H\subseteq E$, 
a community $\community_i$, 
and density constraint $\alpha_i$, 
the potential~\densitypotential\ is defined as 
\begin{equation}
\label{eq:densitypotential1}
\densitypotential(H, \community_i)=\min\set{0,\,\abs{H(\community_i)} - \left\lceil\alpha_i {|\community_i|\choose 2}\right\rceil},
\end{equation}
where, 
$\lceil\cdot\rceil$ denotes the ceiling function.
Note that  
\[
\begin{array}{cccc}
\densitypotential(H,\community_i)<0 & \text{if} & \indicator_{\density\ge\alpha_i}(G(\community_i,H),\,\community_i)=0, &  \text{and} \\
\densitypotential(H,\community_i)=0 & \text{if} & \indicator_{\density\ge\alpha_i}(G(\community_i,H),\,\community_i)=1. &
\end{array}
\]
In other words,  
$\densitypotential$ is negative if the edges $H$ do not satisfy 
the density constraint on $\community_i$, 
and becomes zero as soon as the density constraint is satisfied. 

We also define the \emph{total potential} of a set of edges $H$ 
with respect to the input communities $\communities$ as
\begin{equation}\label{eq:densitypotential}
\densitypotential(H)=\sum_{\community_i\in\communities}\densitypotential(H,\community_i).
\end{equation}

The choices of {\densitygreedy} are guided by the potential function \densitypotential.
The algorithm starts with $E'=\emptyset$
and at each iteration it selects an edge $e\in E \setminus E'$ 
that maximizes the potential difference 
\[
\densitypotential(E'\cup\{e\}) - \densitypotential(E').
\]
The algorithm terminates when it reaches to a set $E$ 
with $\densitypotential(E)=0$, 
indicating that the density constraint is satisfied for all input sets $\community_i\in\communities$.
It can be shown that {\densitygreedy} provides an approximation guarantee.

\begin{proposition}
\label{prop:approximation}
\emph{{\densitygreedy}} is an $\bigO(\log \ell)$-appro\-xi\-ma\-tion 
algorithm for the {\densprb} problem.
\end{proposition}

We obtain
Proposition~\ref{prop:approximation} by using the 
classic result of Wolsey~\cite{Wolsey:1982jx} 
on maximizing motonote and submodular functions.
The key is to show that the potential function 
{\densitypotential} 
is monotone and submodular.

\begin{proposition}
\label{prop:submodularity}
The potential function \densitypotential\ is monotone and submodular.
\end{proposition}

\spara{Running-time analysis.} 
Let $L = \sum_{i=1}^\ell|\community_i|$ 
and $m = |E|$. 
Consider an $m\times L$ table $T$
so that $T[e,i]=1$ if $e\in E(\community_i)$ and $T[e,i]=0$ otherwise.
It is easy to see that {\densitygreedy} can be implemented 
with a constant number of passes for each non-zero entry of $T$, 
giving a running time of $\bigO(L+m\ell)$. 
If we use a sparse implementation for $T$, 
the overall running time becomes $\bigO(L+|T|)$, 
where $|T|$ is the number of non-zero entries of~$T$.

\spara{Adding connectivity constraints.}
Note that solutions to the {\densprb} problem
may be sparse networks in which communities $\community_i$ are dense but {\em disconnected}.
For certain applications we may want do introduce an additional 
connectivity constraint, so that all subgraphs $G(\community_i,E(\community_i))$ are connected.

Combining the two constraints of density and connectivity 
can be handled by a simple adaptation of the greedy algorithm.
In particular we can use a new potential that is the sum of 
the density potential in Equation~\eqref{eq:densitypotential1}
and a potential for connected components. 
This new potential is still monotone and submodular, thus, 
a modified greedy will return a solution that satisfies both density and connectivity constraints
and provides an $\bigO(\log \ell)$-appro\-xi\-ma\-tion guarantee. 

\section{Sparsification with star constraints}
\label{sec:star}

In the second instantiation of the {\setnetworkprb} problem, 
the goal is to find a sparse network $G'$
so that each input community $\community_i\in\communities$ {\em contains a star}, meaning that for every community $\community_i$ subgraph $G(\community_i,E(\community_i))$ has a star spanning subgraph.
The motivation is that a star
has small diameter, as well as
a central vertex that can act as a community leader.
Thus, star-shaped groups have low communication cost
and good coordination properties.

We define the {\starprb} problem as the special case of Problem~{\ref{problem:general}} 
by taking $\Property$ to be the star-containment property. 
We can again show that {\starprb} is a computationally hard problem. 

\begin{proposition}
\label{prop:starhardness}
The \starprb\ problem is \NP-hard.
\end{proposition}

Unfortunately, {\starprb} is not amenable to the same approach we used for {\densprb}, 
hence, we use a completely different set of algorithmic techniques.
Our algorithm, called {\matching} (for {\em ``directed stars to stars''}),
is based on solving a \emph{directed version} of {\starprb}, 
which is formally defined as follows. 

\begin{problem}[\distarprb]
\label{problem:distar}
Consider a directed network $N=(V,A)$, and a set of $\ell$ communities
$\communities = \{\community_1,\ldots, \community_\ell\}$. 
We want to find a sparse directed network $N' = (V, A')$ 
with $A'\subseteq A$, 
such that, 
the number of edges $\abs{A'}$ is minimized
and for each community 
$\community_i\in\communities$ there is a central vertex $c_i \in \community_i$ 
with $(c_i \rightarrow x)\in A'$ for all $x \in \community_i\setminus\{c_i\}$.
\end{problem}

In Problem~\ref{problem:distar}, 
the original network $N$ and the sparsifier $N'$ are both directed. 
Note, however, that {\distarprb} can be also defined 
with an undirected network $G$ as input:
just create a directed version of $G$, 
by considering each edge in both directions. 
Thus we can consider that {\starprb} and {\distarprb} take the same input. 
In this case, it is easy to verify the following observation.

\begin{observation}
\label{observation:relationship}
If $G^\ast = (V,E^\ast)$ is the optimal solution for {\starprb}, 
and $N^\ast=(V,D^\ast)$ is the optimal solution for {\distarprb},  
for the same input, 
then \[\abs{E^{\ast}} \leq \abs{D^\ast}\leq 2 \abs{E^{\ast}}.\]
\end{observation}

As with {\starprb}, 
the {\distarprb} problem is \NP-hard.\footnote{The 
proof of Proposition~\ref{prop:starhardness} can be modified slightly to show 
\NP-hardness for {\distarprb}.} 
Our approach is to solve {\distarprb} 
and use the directed sparsifier $N'=(V,D')$ to 
obtain a solution for {\starprb}. 
Observation~\ref{observation:relationship}
guarantees that the solution we obtain for {\starprb} in this way
is not far from the optimal.
In the next paragraph, we describe 
how to obtain an approximation algorithm for the {\distarprb} problem.

\spara{Solving {\distarprb}.}
First, we observe that {\distarprb} can be viewed as a {\hyperprb} problem, 
which is defined as follows:  
we are given a set of elements $X$, 
a collection of \emph{hyperedges} $\mathcal{H} = \{H_1,\dots H_t\}$,
$H_i \subseteq X$, 
and a score function $\funcdef{c}{\mathcal{H}}{\Reals}$. 
We seek to find a set of disjoint hyperedges 
$\mathcal{I}\subseteq\mathcal{H}$ 
maximizing the total score 
$\sum_{H_i\in \mathcal{I}} c(H_i)$.

The mapping from {\distarprb} to {\hyperprb} is done as follows:
We set $X$ to $\communities$.  
Given a subset $H \subseteq \communities$, 
let us define $\intersection(H)$ and $\union(H)$ 
to be the intersection and union of members in $H$,
respectively. Let us construct a set of hyperedges as
\[
\mathcal{H} = \left\{H \mid H \subseteq \communities \text{ and } \intersection(H) \ne \emptyset\right\},
\]
and assign scores
\[
c(H) = 1 - \abs{H} + \sum_{v \in \union(H)} \abs{\set{i \mid v \in \community_i \in H}} - 1.
\]
Note that
if represented na\"{i}vely, 
the resulting hypergraph can be of exponential size. 
However, this problem can be avoided easily by 
an implicit representation of the hypergraph. 
We can now show that the optimal solution 
to the transformed instance of {\hyperprb} 
can be used to obtain an optimal solution to {\distarprb}.

\begin{proposition}
\label{prop:optimal}
Let $\mathcal{I}$ be the optimal solution
for the {\hyperprb} problem instance. 
Let $D_\mathcal{I}$ be the union of directed stars, 
each star having a center in $\intersection(H)$
and directed edges towards vertices in $\union(H)$ for every $H \in \mathcal{I}$. 
If $D^\ast$ is the optimal solution to the 
{\distarprb} problem, then 
\[
	 \abs{D^\ast} = \abs{D_\mathcal{I}}=\sum_{\community_i\in\communities}(|\community_i| - 1)  - \sum_{H \in \mathcal{I}} c(H).
\]
\end{proposition}

Consider a greedy algorithm, {\hypergreedy}, 
which constructs a solution to {\hyperprb}
by 
iteratively adding hyperedges to~$\mathcal{J}$
so that it maximizes $\sum_{H_i\in \mathcal{J}} c(H_i)$,
while keeping $\mathcal{J}$ disjoint. 
As shown in the supplementary material, 
{\hypergreedy} is a $k$-factor approximation algorithm
for the {\hyperprb} problem.
The proof is based on the concept of $k$-extensible systems~\cite{mestre2006greedy}. 

\begin{proposition}\label{prop:kapprox}
Let $\mathcal{J}$ be the resulting set of hyperedges given by
the {\em\hypergreedy} algorithm, 
and let $\mathcal{I}$ be the optimal solution for {\hyperprb}.
Then,
\[
	\sum_{H \in \mathcal{I}} c(H) \leq k \sum_{H \in \mathcal{J}} c(H), \quad\text{where}\quad k = \max_{H \in \mathcal{H}} \abs{H}.
\]
\end{proposition}

Propositions~\ref{prop:optimal} and~\ref{prop:kapprox} imply the following.

\begin{corollary}
\label{cor:dirbound}
Let $D^\ast$ be the optimal solution of the {\distarprb} problem.
Let $\mathcal{J}$ be the greedy solution to the corresponding {\hyperprb} problem,
and let $D_\mathcal{J}$ be the corresponding edges 
(obtained as described in  Proposition~\ref{prop:optimal}). 
Then,
\[
	\abs{D_\mathcal{J}} \leq \frac{k - 1}{k} C + \frac{1}{k}\abs{D^\ast},
\quad\text{where}\quad C = \sum_{\community_i\in\communities} \left(\abs{\community_i} - 1\right)
\]
and $k$ is the maximum number of sets in $\communities$ that have a non-empty intersection.
\end{corollary}

\spara{Putting the pieces together.}
The pseudo-code of {\matching} is shown in Algorithm~\ref{algo:matching}.  
In the first step, the algorithm invokes {\hypergreedy}
and obtains a solution to the {\hyperprb} problem.
This solution is then translated into a solution to the {\distarprb} problem
(function $\texttt{H2D}$).
Finally, the solution to {\distarprb} is translated into a solution to the {\starprb}
by transforming each directed edge in $D_\mathcal{J}$ 
into an undirected edge and removing duplicates
(function $\texttt{D2E}$).
We have the following result.
\begin{algorithm}[t]
	\caption{\label{algo:matching}The {\matching} algorithm for {\starprb}.}
\KwIn{$G_0=(V,E_0)$ and $\communities = \{\community_1,\ldots , \community_\ell\}$}
\KwOut {Graph $G=(V,E)$ such that $E(\community_i)$ contains a star, for all $i\in 1,\ldots ,\ell$.} 
	$\mathcal{J}=$ {\hypergreedy}$(\communities)$; \\
	$D_\mathcal{J} = \texttt{H2D}(\mathcal{J})$;  \\
	$E = \texttt{D2E}(D_\mathcal{J})$; \\
	\Return $G=(V,E)$;
\end{algorithm}
\begin{proposition}
\label{prop:undirbound}
Let $E^\ast$ be the optimal solution of the {\starprb} problem.
Let $E$ be the output of the {\em\matching} algorithm. Then,
\[
	\abs{E} \leq \frac{k - 1}{k} C + \frac{2}{k}\abs{E^\ast}, \quad\text{where}\quad C = \sum_{\community_i\in\communities} (\abs{\community_i} - 1)
\]
and $k$ is the maximum number of sets in
$\communities$ that have a non-empty intersection.
\end{proposition}

\spara{Running time.}
The running time of the {\matching}  algorithm
is dominated by {\hypergreedy}.
The other two steps (lines 2 and 3 in Algorithm~\ref{algo:matching})
require linear time with respect to $|V|$.
{\hypergreedy} can be implemented with a priority queue. 
In each step we need to extract the maximum-weight hyperedge, 
and update all intersecting hyperedges by removing any common sets. 
The number of maximal hyperedges in  $\mathcal{H}$ is at most $|V|$  
(one for each vertex~$v$), 
and assuming that the maximum number of intersecting hyperedges is bounded by $c$, 
the total running time of the algorithm is $\bigO(|V|\ell\log|V| + \ell\sum \abs{\community_i})$.

\section{Experimental evaluation}
\label{sec:exp}

In this section we discuss 
the empirical performance of our methods.
Our experimental study is guided
by the following questions.

\squishlist
\item[{\bf Q1.}]
How do our algorithms compare against competitive sparsification baselines
that also aim at preserving the community structure?

\item[{\bf Q2.}]
How well is the structure of the sparsified network preserved
compared to the structure of the original network?


\item[{\bf Q3.}]
What are specific case studies that
support the motivation of our problem formulation?
\squishend

\noindent
We note that the implementation of the algorithms and all datasets used 
will be made publicly available.



\begin{table*}[t]
	\setlength{\tabcolsep}{0pt}
	\caption{Network characteristics. $|V|$: number of nodes; $|\edges|$:
	number of edges in the underlying network; $|\indedges|$: the number of edges
	induced by communities; $C$: the number of connected components; $\ell$: number
	of sets (communities); $\avg(\alpha_0)$: average density of the ground-truth
	subgraphs induced by the communities; $s_{\min}$, $s_{\avg}$: minimum and average
	set size; $t_{\max}$, $t_{\avg}$: maximum and average participation of a
	vertex to a set.}
	\centering
	{\small\begin{tabular*}{\textwidth}{@{\extracolsep{\fill}}l l rrrrrrrrr}
		\toprule 
		Dataset & $|V|$ &$|\edges|$& $|\indedges|$ & $C$ & $\ell$&$\avg(\alpha_0)$ & $s_{\min}$&$s_{\avg}$& $t_{\max}$&$t_{\avg}$\\		
		\midrule
		{\amazon} (\dataset{1}) & 10001 & 25129 & 17735 & 7 & 11390 & 0.769 & 2 & 3.52 & 20 & 4.01\\
		{\DBLP} (\dataset{2}) & 10001 & 27687 & 22264 & 1 & 1767 & 0.581 & 6 & 7.46 & 10 & 1.31\\
		{\youtube} (\dataset{3}) & 10002 & 72215 & 15445 & 1 & 5323 & 0.698 & 2 & 4.02 & 82 & 2.14\\		
		{\KDD} (\dataset{4}) & 2891 & 11208 & 5521 & 58 & 8103 & 0.178 & 2 & 31.16 & 1288 & 137.00\\
		{\ICDM} (\dataset{5})  & 3140 & 10689 & 5079 & 112 & 8623 & 0.183 & 2 & 32.46 & 1339 & 139.10\\
		{\FBcirc} (\dataset{6}) & 4039 & 88234 & 55896 & 1 & 191 & 0.640 & 2 & 23.15 & 44 & 1.53\\
		{\FBfeat} (\dataset{7}) & 4039 & 88234 & 84789 & 1 & 1245 & 0.557 & 2 & 29.78 & 37 & 9.21\\
		{\lastFMart} (\dataset{8}) & 1892 & 12717 & 5253 & 20 & 7968 & 0.047 & 2 & 8.29 & 1147 & 36.73\\
		{\lastFMtag} (\dataset{9}) & 1892 & 12717 & 7390 & 20 & 2064 & 0.053 & 2 & 13.60& 50 & 15.43\\
		{\DBbook} (\dataset{10}) & 1861 & 7664 & 1213 & 62 & 8337 & 0.069 & 2 & 3.34 & 58 & 15.32\\
		{\DBtag} (\dataset{11}) & 1861 & 7664 & 6293 & 62 & 14539 & 0.032 & 2 & 13.79 & 658 & 107.60\\	
		{\birds} (\dataset{12}) & 1052 & 44812 & 44812 & 1 & 49578 & 1.0 & 2 & 6.03 & 938 & 284.30\\
		{\cocktails} (\dataset{13}) & 334 & 3619 & 3619 & 1 & 1600 & 1.0 & 2 & 3.73 & 427 & 17.89\\			
		\bottomrule
	\end{tabular*}}
	\label{tab:stats}
\end{table*}

\spara{Datasets.}
We use $13$ datasets (\dataset{1}--\dataset{13}); each dataset
consists of a network $G=(V,E)$ and a set of communities {\communities}. 
We describe these datasets below, 
while their basic characteristics are shown in 
Table~\ref{tab:stats}.

\smallskip
\noindent
$\bullet$ {\KDD} and {\ICDM} are subgraphs of the DBLP co-authorship network. 
Edges represent co-authorships between authors.
Communities are formed by keywords that appear in paper abstracts.

\smallskip
\noindent
$\bullet$ \FBcirc and \FBfeat are Facebook ego-networks available 
at the SNAP repository.\footnote{\url{snap.stanford.edu/data/egonets-Facebook.html}} 
In {\FBcirc} the communities are social-circles of users.
In {\FBfeat} communities are formed by user profile features.

\smallskip
\noindent
$\bullet$ 
{\lastFMart} and {\lastFMtag} are friendship networks of last.fm 
users.\footnote{\url{grouplens.org/datasets/hetrec-2011/}}
A community in {\lastFMart} and {\lastFMtag} is formed by users who listen to the
same artist and genre, respectively.
 
\smallskip
\noindent
$\bullet$ 
{\DBbook} and {\DBtag} 
are friendship networks of Delicious users.\footnote{\url{www.delicious.com}}
A community in \DBbook and \DBtag is formed by users who use the same 
bookmark and keyword, respectively.

\smallskip
Additionally, 
we use 
SNAP datasets with ground-truth communities. 
To have more focused groups, 
we only keep communities with size less than~10.
To avoid having disjoint communities, 
we start from a small number of seed communities and iteratively add other communities
that intersect at least one of the already selected.
We stop when the number of vertices reaches $10\,$K.
In this way we construct the following datasets:

\smallskip
\noindent
$\bullet$ \amazon: 
Edges in this network represent pairs of frequently co-purchased products.
Communities represent product categories as provided by Amazon.

\smallskip
\noindent
$\bullet$ \DBLP:
This is also a co-authorship network.
Communities are defined by publication venues.

\smallskip
\noindent
$\bullet$ \youtube: 
This is a social network of Youtube users.
Communities consist of user groups created by users.

\smallskip
For the case studies we use the following datasets.

\smallskip
\noindent
$\bullet$ \cocktails:\footnote{\url{www-kd.iai.uni-bonn.de/InVis}}
Vertices represent drink ingredients
and communities correspond to ingredients appearing in cocktail recipes. 
The \cocktails dataset does not have a ground-truth network.

\smallskip
\noindent
$\bullet$ \birds: 
This dataset consists of group sightings of \emph{Parus Major} 
(great tit)~\cite{Farine150057}. 
The dataset also contains gender, age,
and immigrant status of individual birds.

\spara{Experimental setup.}
All datasets consist of a graph $G=(V,\edges)$ and a set of communities
$\communities=\{\community_1,\ldots ,\community_\ell\}$. 
The output of our algorithms is a sparsified graph
$\sparseG=(V, \sparseedges)$. 
Clearly, any reasonable sparsification algorithm 
would include in $\sparseedges$ only edges that belong in at least one of the graphs
$G[\community_i]=(\community_i, E[\community_i])$.
Accordingly, we define $\indedges$ to be the set of edges belonging in at least one such 
subgraph: $\indedges = \cup_{i=1\ldots \ell}E[\community_i]$.

{\densprb} requires a density threshold $\alpha_i$ for each community 
$\community_i\in\communities$.
We set this parameter pro\-por\-tio\-nal to the density ${D_i}$ of $G[\community_i]$.
We experiment with 
$\alpha_i=\epsilon\,{D_i}$, for $\epsilon=0.5$, $0.7$, and $0.9$.

{\starprb} aims to find a star in every community $G[\community_i]$. 
If no star is contained in $G[\community_i]$
then the community $\community_i$ is discarded.

\spara{Baseline.} 
We compare our algorithms with a sparsification method, 
proposed by Satuluri et al.~\cite{satuluri2011local} to enhance community detection, 
and shown in a recent study 
by Lindner et al.~\cite{lindner2015structure}
to outperform its competitors 
and to preserve well network cohesion.  
The algorithm, which we denote by \LS, 
considers local similarity of vertex neighborhoods.
The reader should keep in mind that
{\LS} is not optimized for the problems we define in this paper, 
and in fact, it does not use the communities {\communities} as input. 
Nonetheless we present {\LS} in our evaluation, 
as it is a state-of-the-art sparsification method
that aims to preserve community structure.


 
\setlength{\tabcolsep}{0pt}
\newlength{\figlength}
\setlength{\figlength}{0.25\textwidth}

\begin{figure}[t]
\includegraphics[width=\columnwidth]{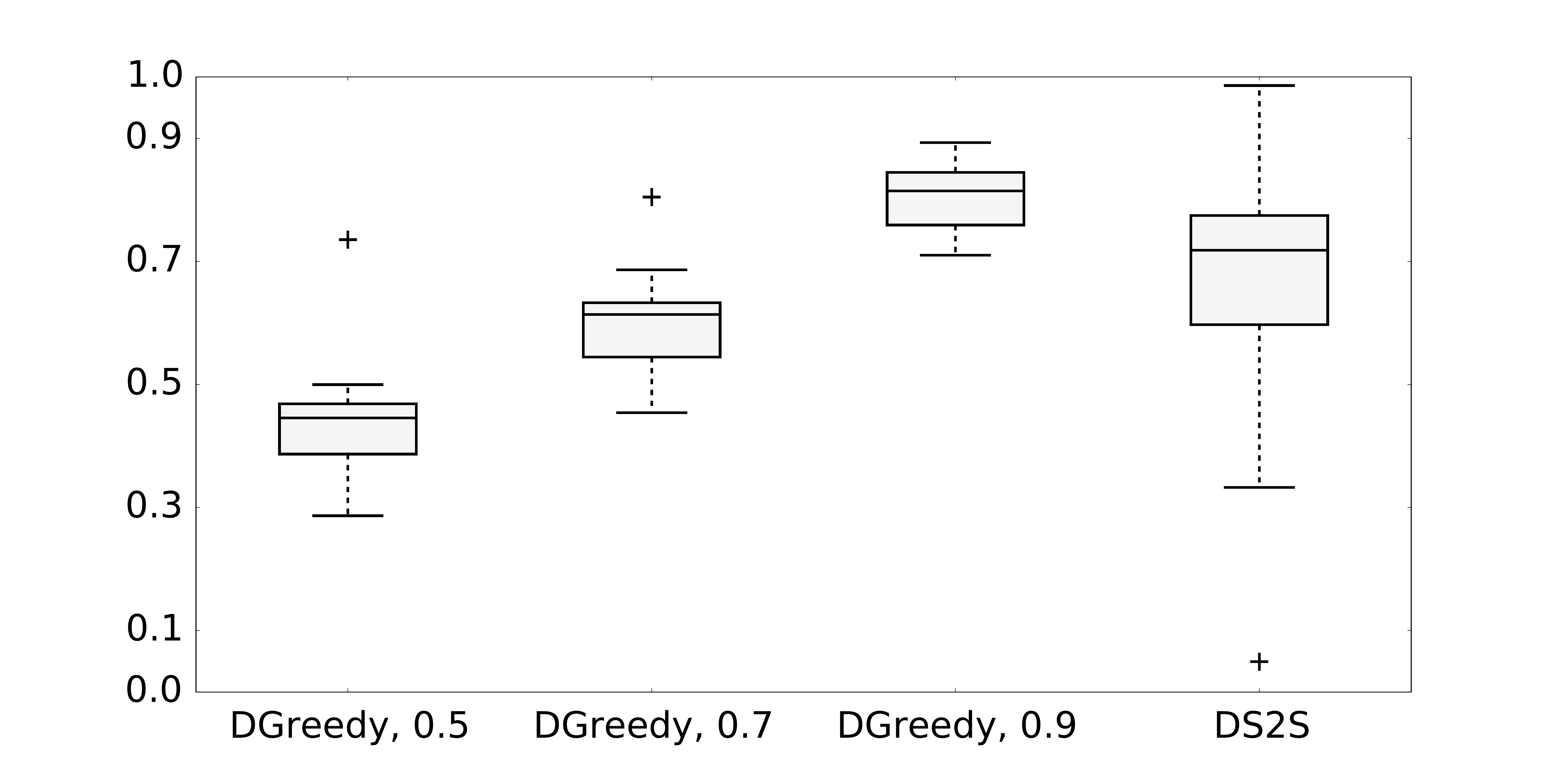}
\caption{\label{figure:sparsification-level}
Amount of sparsification 
shown as distribution of the values $\ratio=|\sparseedges|/|\indedges|$
over the different datasets for each problem setting.}
\end{figure}
\begin{figure*}[t]
	\begin{center}
		\setlength{\tabcolsep}{0pt}
		\setlength{\figlength}{0.25\textwidth}
		\begin{tabular}{c@{\hspace*{1mm}}*{4}{c}}
			Star & Density, $\epsilon=0.5$ & Density, $\epsilon=0.7$ & Density, $\epsilon=0.9$\\
			\includegraphics[width=\figlength]{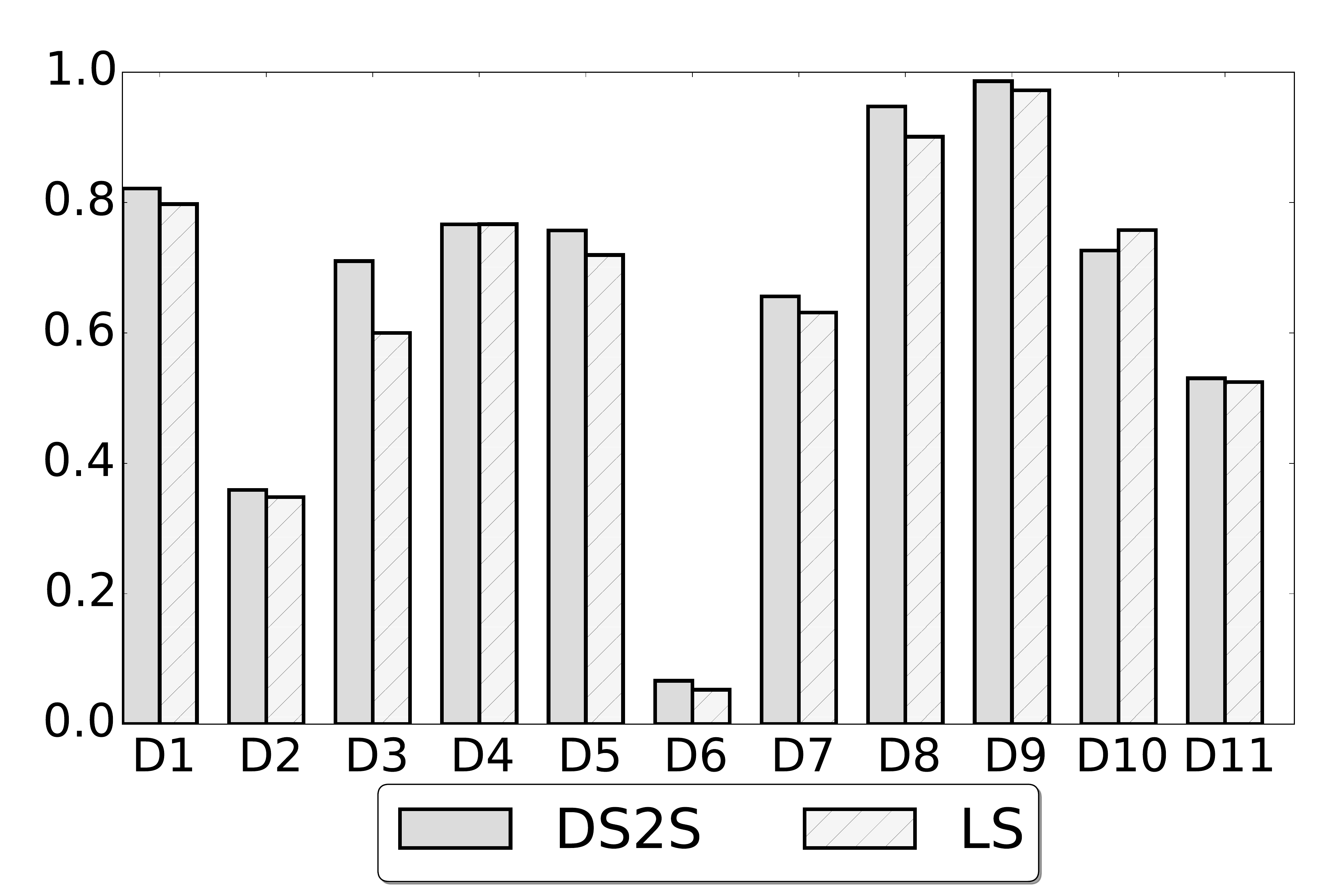}             
			&\includegraphics[width=\figlength]{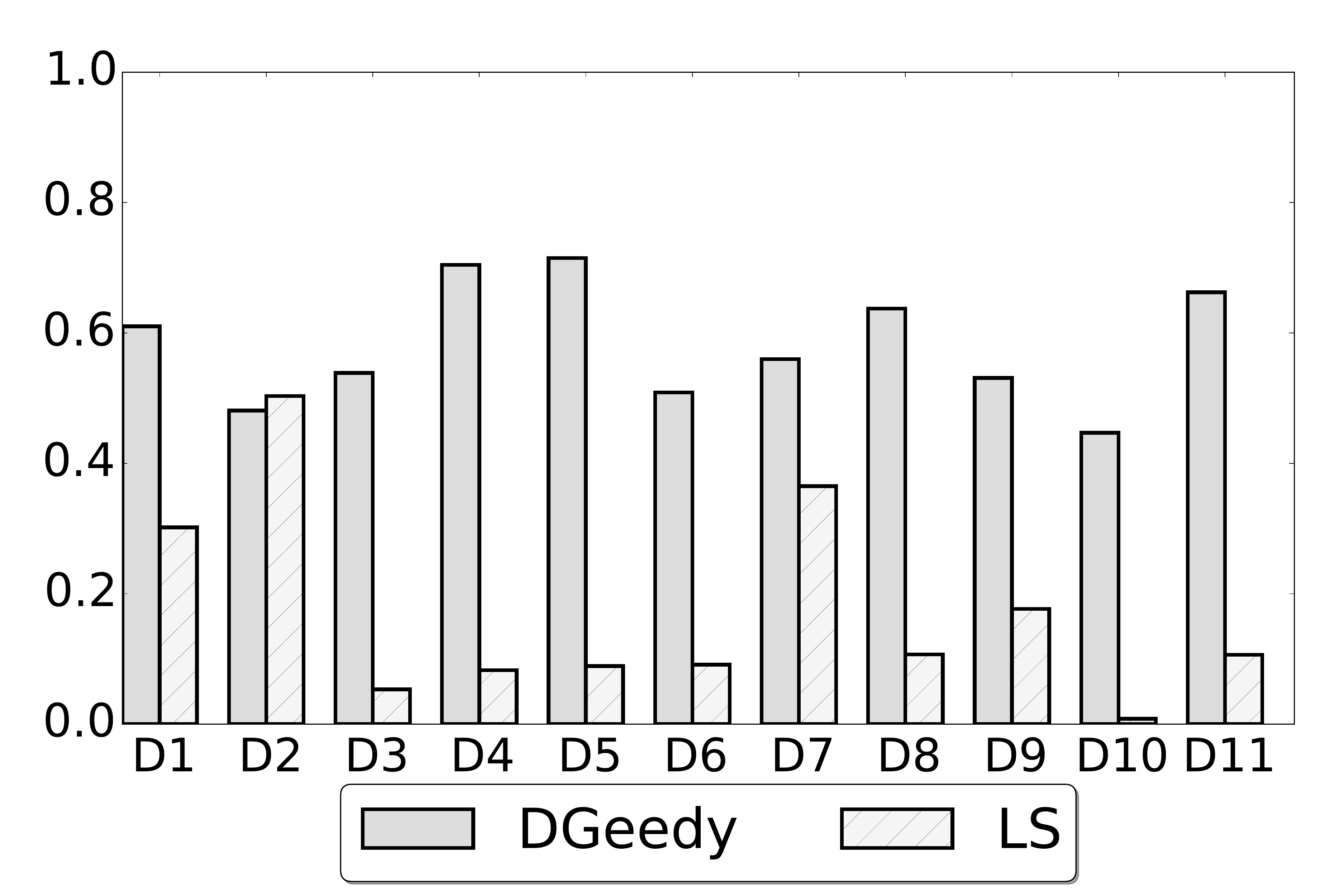}
			&\includegraphics[width=\figlength]{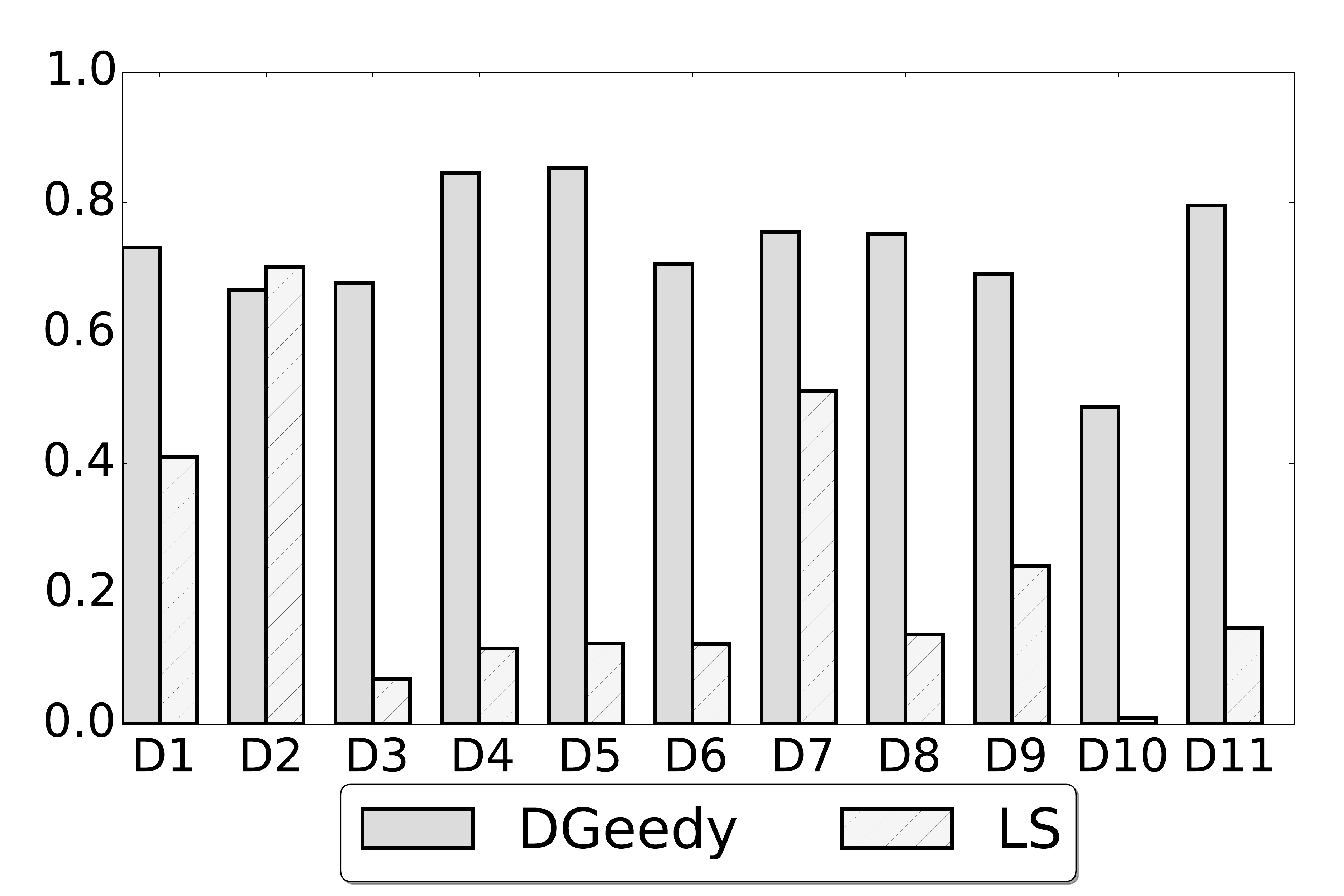} 
			&\includegraphics[width=\figlength]{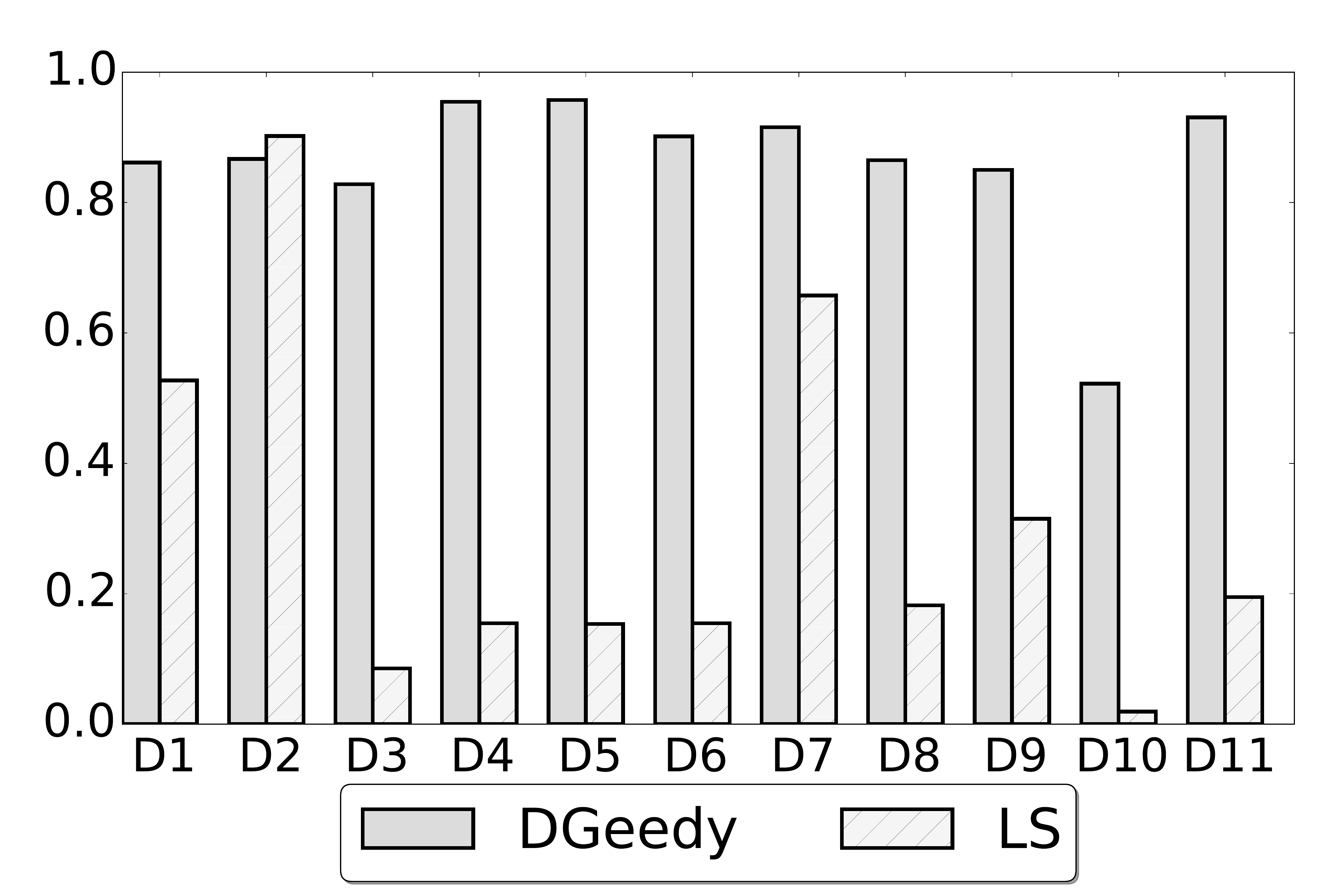}
		\end{tabular}
	\end{center}
	\caption{Relative degree $\reldegree$ within communities, on the sparsified graphs {\sparseG} produced by $\densitygreedy$ and baseline $\LS$ for datasets \dataset{1}--\dataset{11}. Measure $\reldegree$ is defined as average degree within community in the sparsified graph divided by average degree within community in the input graph. Larger values of $\reldegree$ correspond to better preserved community sets.}
	\label{fig:avgdegreeinset_relative}
\end{figure*}

\begin{figure*}[t]
	\begin{center}
		\setlength{\tabcolsep}{0pt}
		\setlength{\figlength}{0.25\textwidth}
		\begin{tabular}{c@{\hspace*{1mm}}*{4}{c}}
			Star & Density, $\epsilon=0.5$ & Density, $\epsilon=0.7$ & Density, $\epsilon=0.9$\\
			\includegraphics[width=\figlength]{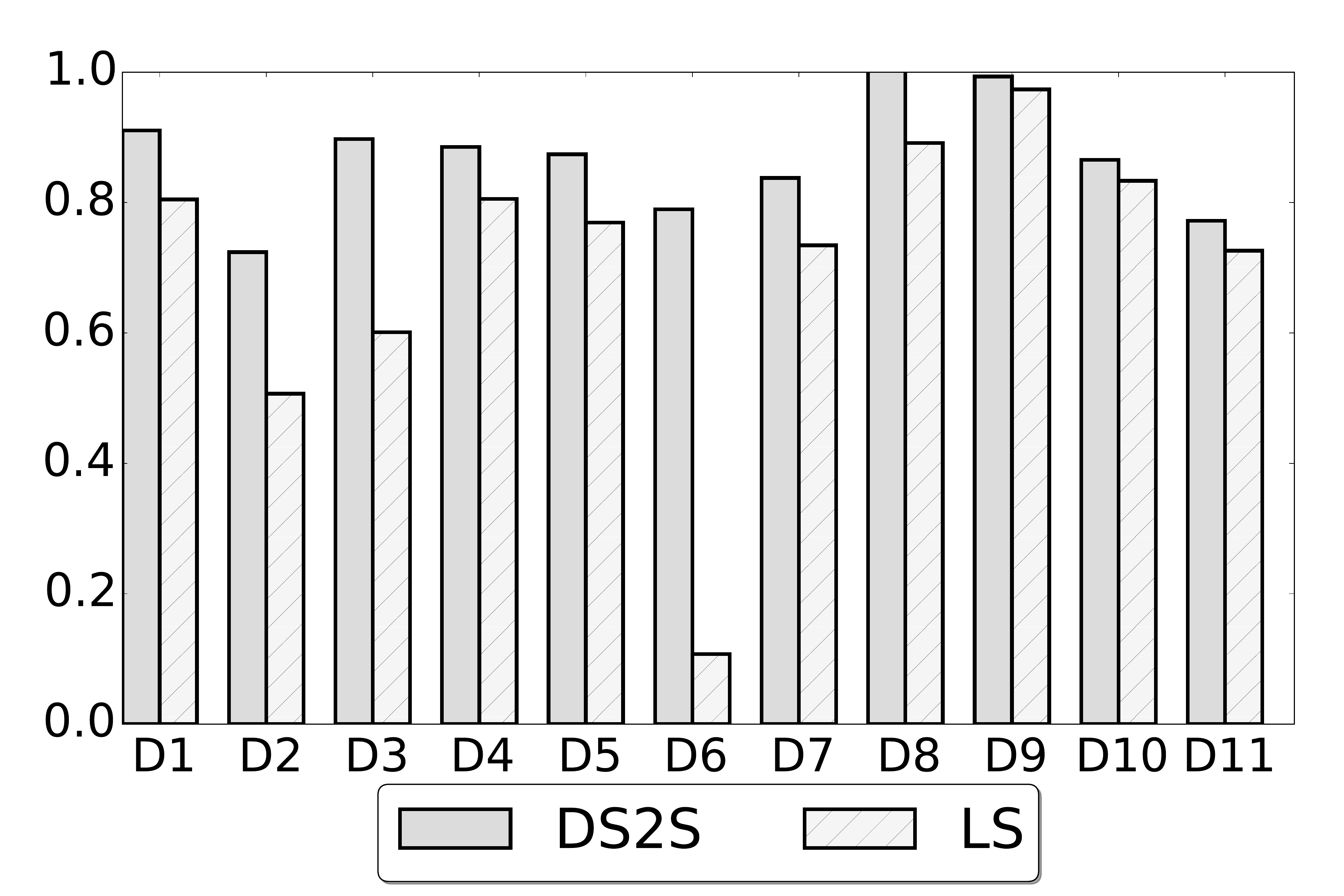}              
			&\includegraphics[width=\figlength]{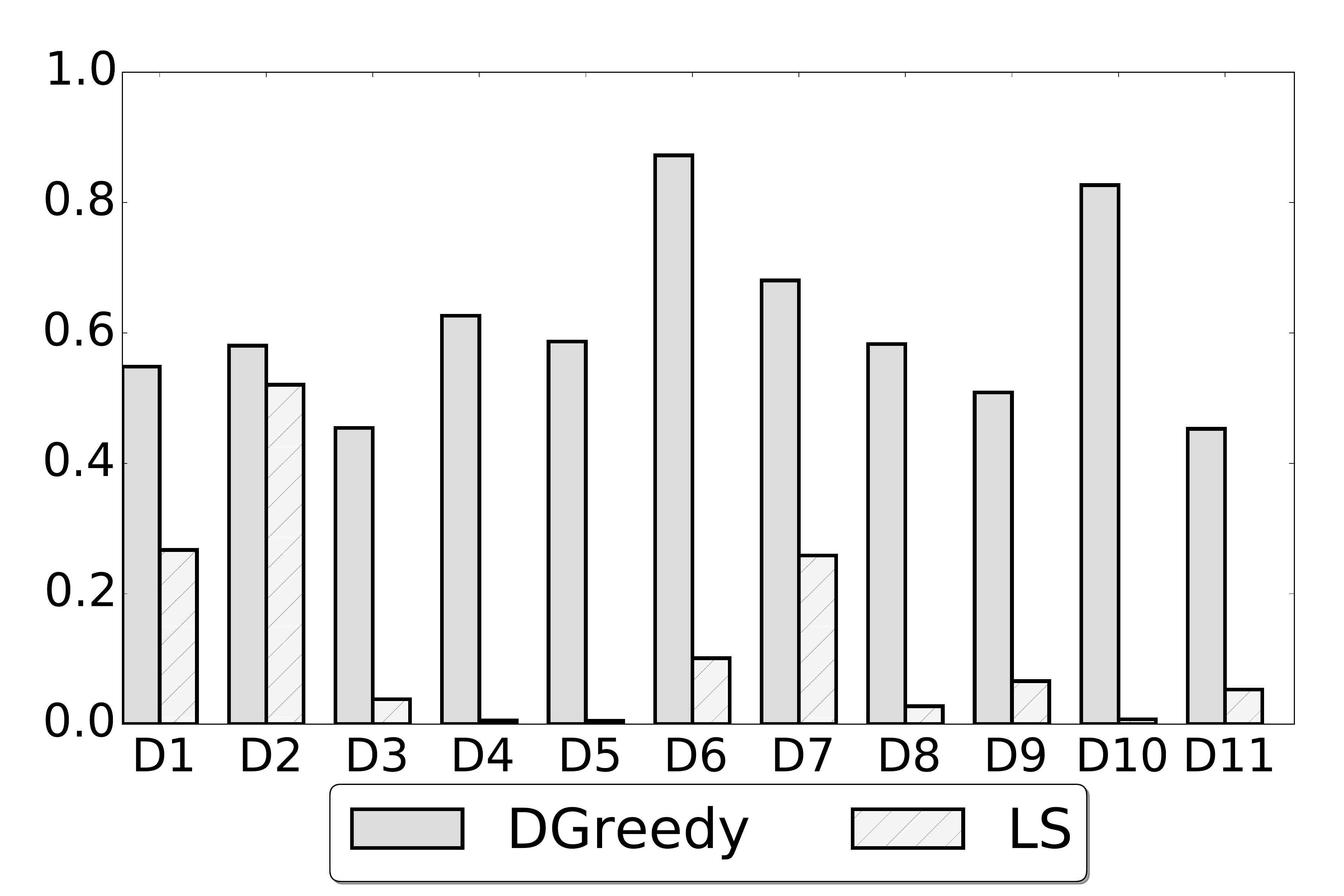}
			&\includegraphics[width=\figlength]{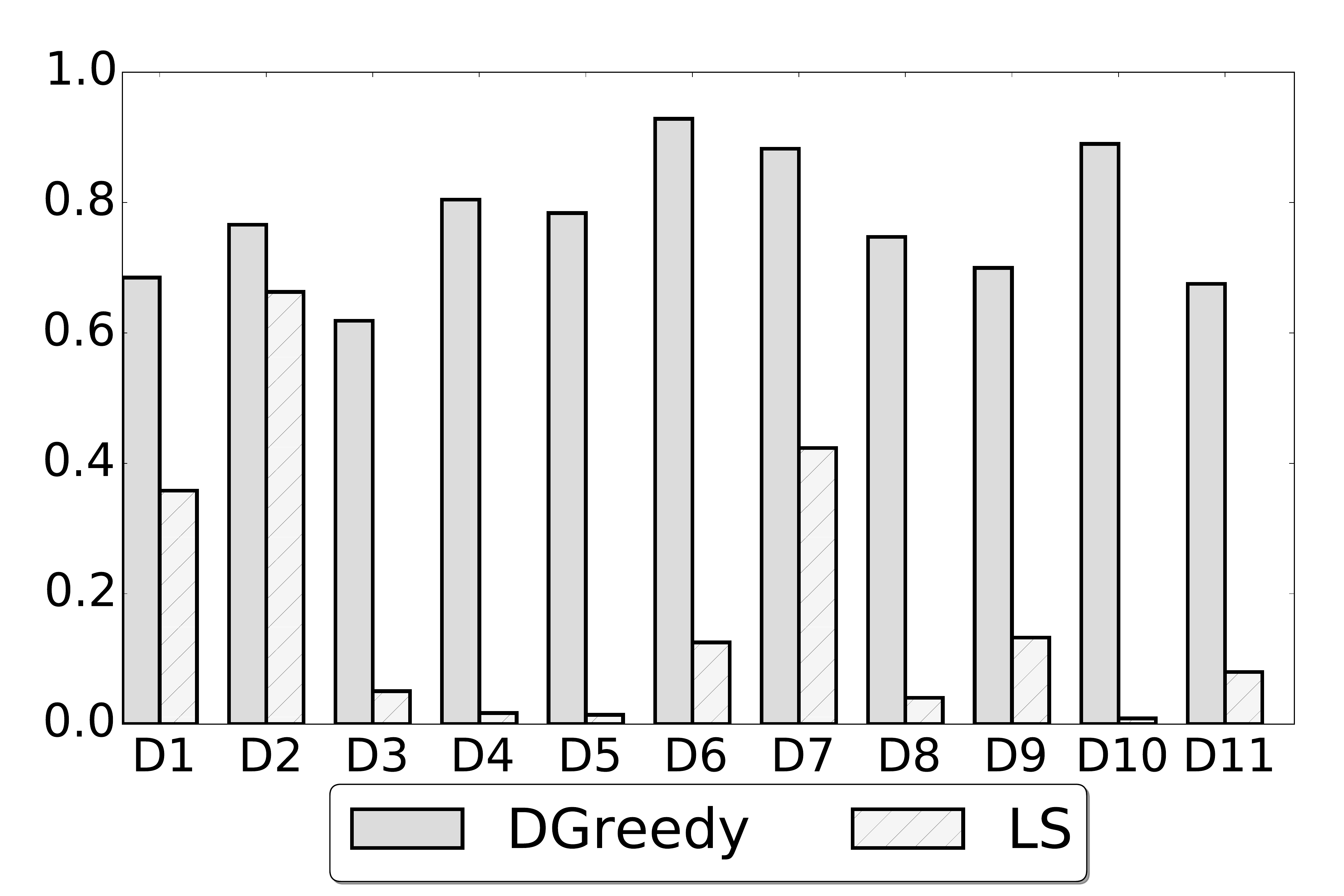} 
			&\includegraphics[width=\figlength]{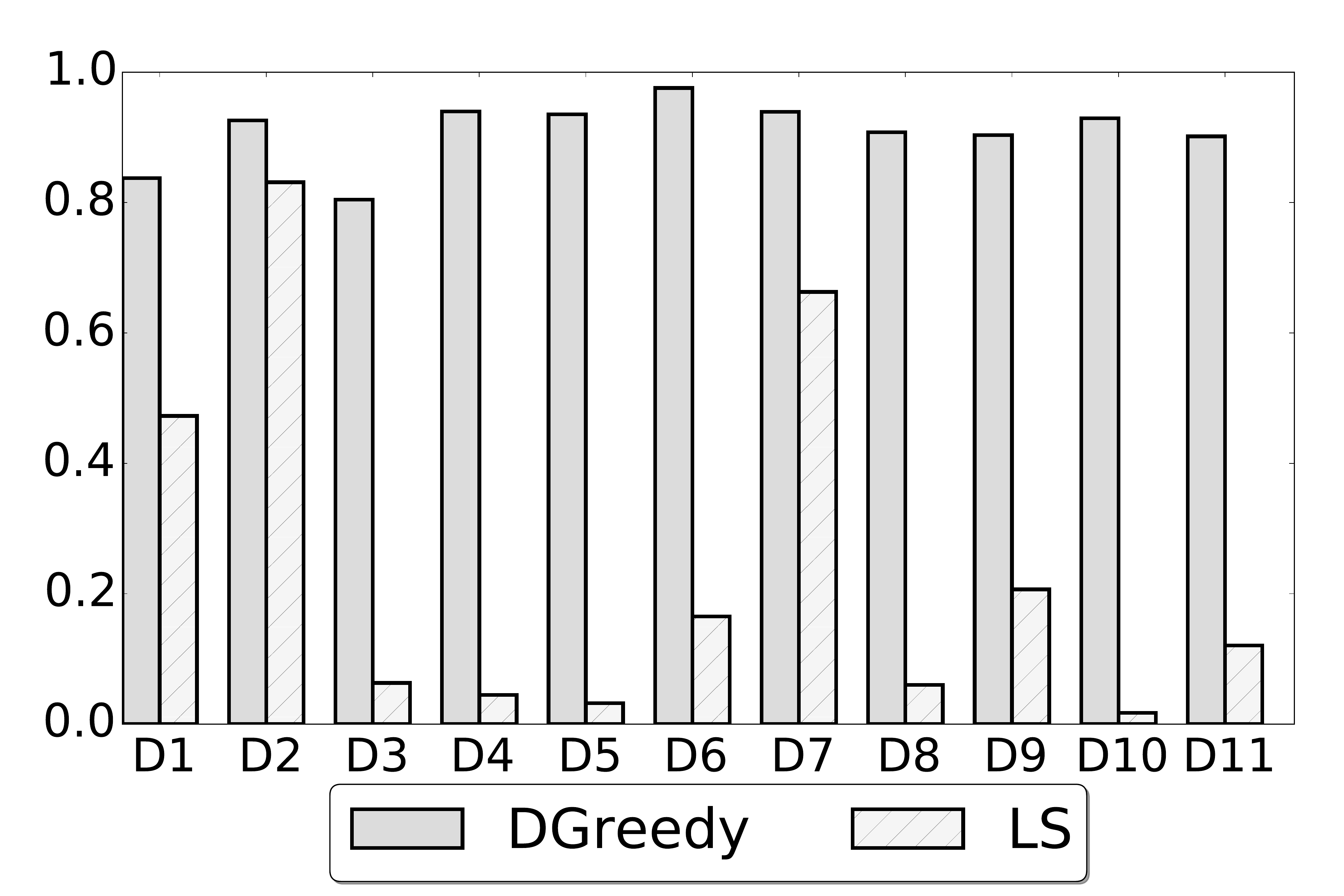}
		\end{tabular}
	\end{center}
	\caption{Relative paths length $\relSPlen$, on the sparsified graphs {\sparseG} produced by $\densitygreedy$ and baseline $\LS$ for datasets \dataset{1}--\dataset{11}. Measure $\relSPlen$ is defined as harmonic mean of within communities shortest-path lengths in the input graph divided by corresponding harmonic mean in the sparsified graph. Again, larger values of $\relSPlen$ correspond to better preserved community sets.}
	\label{fig:Hspinset_relative}
\end{figure*}

\spara{Amount of sparsification.}
Starting with input network
$\indgraph = (V, \indedges)$ and communities {\communities}
we compute a sparsified network {\sparseG}, 
for datasets \dataset{1}--\dataset{11}. 
We solve the {\starprb} problem using the {\matching}
and the  {\densprb} problem using the {\densitygreedy} algorithm
for the three different values of $\epsilon$ we consider.
We quantify the amount of sparsification 
by the ratio $\ratio=|\sparseedges|/|\indedges|$, 
which takes values from 0 to 1, and  
the smaller the value of $\ratio$ the larger the amount of quantification. 
Boxplots with the distribution of the values of $\ratio$
for the four different cases 
({\starprb}, and {\densprb} with $\epsilon=0.5, 0.7, 0.9$)
are shown in Figure~\ref{figure:sparsification-level}.

The {\LS} baseline takes as parameter the number of edges in the sparsified network. 
Thus, for each problem instance 
we ensure that the sparsified network obtained by {\LS}
have the same number of edges
(up to a 0.05 error margin controlled by \LS)
as the sparsified networks obtained by our methods in the corresponding instances.

\spara{Properties of sparsified networks.} 
We start by considering our first two evaluation questions
{\bf Q1} and {\bf Q2}. For this, we
compare our methods with a competitive baseline (\LS), and
we
quantify the amount structure preservation in the sparsified network.

Recall that for an input network
$\indgraph = (V,\indedges)$ and communities {\communities}
we compute a sparsified network~{\sparseG}, 
for datasets \dataset{1}--\dataset{11}.
We {\em compare} the networks~{\sparseG} and {\indgraph}
by computing the {\em average degree} and the 
{\em average shortest-path length} within the communities 
{\communities}.\footnote{The average shortest-path length is 
estimated using the harmonic mean, 
which is robust to disconnected components.}
The goal is to test whether within-communities statistics
in the sparsified network are close to those in the original network. 
The results for average degree and average shortest path are shown in 
Figures~\ref{fig:avgdegreeinset_relative} and~\ref{fig:Hspinset_relative}, 
respectively.
The leftmost panel in each figure shows the results for the 
the {\starprb} problem, 
while the other three panels show the results 
for the {\densprb} problem, 
for the three different values of $\epsilon$ we consider.

As expected, in the sparsified network, 
average degrees decrease and short-path lengths increase.
For {\densprb}, as $\epsilon$ increases,
both average distance and average shortest-path length in the 
sparsified network come closer to their counterparts in the input network. 
For the {\starprb} problem the {\LS} baseline is competitive
and in most cases it produces networks whose statistics
are close to the ones of the networks produced by {\matching}.
However, for the {\starprb} problem, 
the {\LS} baseline does not do a particularly good job 
in preserving community structure.

Overall this experiment reinforces our understanding that while 
sparsification is effective with respect to reducing the number of edges, the
properties of the communities in the sparsified network
resemble respective properties in the input network.

\spara{Running time.} 
For all reported datasets the total running time of \matching\ is under 1 second, while
\densitygreedy\ completes in under 5 minutes. 
The experiments are conducted on a machine with Intel Xeon 3.30GHz and 15.6GiB of memory.


\begin{table}[t]
	\caption{Top-10 star centers, discovered by \matching\ algorithm on \cocktails\ dataset. The centers are ordered by the discovered order, with the number of sets a center covers in parentheses.}
	\centering
	{\small\begin{tabular*}{7cm}{@{\extracolsep{\fill}}ll}
		\toprule
		\texttt{vodka} (202) & \texttt{gin} (86)\\
		\texttt{orange j.} (118) & \texttt{amaretto} (85)\\
		\texttt{pineapple j.} (86) & \texttt{light rum} (58)\\
		\texttt{bailey's} (78) & \texttt{kahlua} (58)\\
		\texttt{tequila} (81)& \texttt{blue curacao} (50)\\
		\bottomrule
	\end{tabular*}}
	\label{tab:cocktails}
\end{table}

\spara{Case studies.}
To address evaluation question {\bf Q3}
we conduct two case studies, 
one presented here and one in the supplementary material.
In both cases there is no underlying network, 
so they can be considered instances of the {\em network design} problem.

\newpage
\noindent
\emph{Cocktails case study.}  
In this case the input communities are defined by the 
ingredients of each cocktail recipe. 
%
We first run the {\matching} algorithm on the input sets, 
and we obtain network $\sparseG$ with $1\,593$ edges, 
that is, around 44\,\% of the edges of $G$, giving an average degree of $9.5$.  
The first ten star centers, in the order selected by the {\matching} algorithm are shown in 
Table~\ref{tab:cocktails}. 
The table also shows the number of cocktails that each ingredient serves as a star.
We see that the algorithm selects popular ingredients that form the basis for many cocktails.
A snippet of the reconstructed network
is shown in Figure~\ref{figure:snippet}. 

In order to compare the outputs of {\matching} ($\sparseG_1$)
and {\densitygreedy}, we ran the latter with  
density parameter $\alpha=0.65$.
For this value of $\alpha$ we get $\sparseG_2$ having $1\,420$ edges,
so that we can have a more meaningful comparison of $\sparseG_1$ and $\sparseG_2$.
In Figure~\ref{figure:cocktail-degrees}
we depict the degree of each vertex 
in the two reconstructed networks,  $\sparseG_1$ and $\sparseG_2$, 
as a function of their degree in the underlying network $G$.
From the figure, we observe that \emph{some} of the unusually high-degree vertices in $G$ 
maintain
their high degree in $\sparseG_1$;  these are probably the vertices
that the {\matching} algorithm decides that they serve as star centers.
On the other hand, there are other high-degree vertices in $G$ 
that lose their high-degree status in $\sparseG_1$;
these are the vertices that the {\matching} algorithm did not use as star centers.
On the other hand, the {\densitygreedy} algorithm 
sparsifies the feasibility network 
much more uniformly and vertices 
maintain their relative degree in $\sparseG_2$.
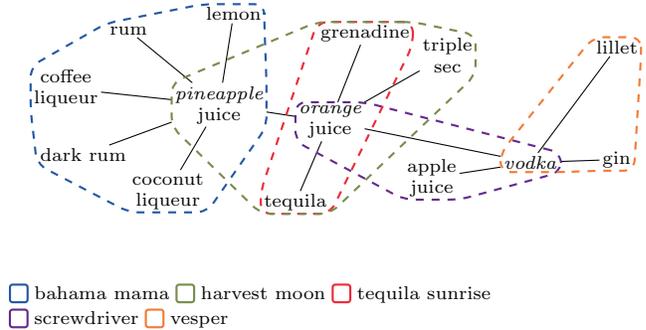
\begin{figure}[t]
\begin{tikzpicture}[yscale = 0.75, xscale = 0.82, rotate = 90, baseline = 32]
\tikzstyle{exnode} = [inner sep = 1pt, font = \scriptsize, align = center]
\tikzstyle{fitnode} = [inner sep = 1pt, rounded corners, dashed, thick, gray]
\tikzstyle{fitline} = [rounded corners, dashed, thick, gray]
\tikzstyle{exedge} = []
\tikzstyle{labnode} = [draw, rounded corners = 1pt, thick]

\definecolor{yafcolor1}{rgb}{0.4, 0.165, 0.553}
\definecolor{yafcolor2}{rgb}{0.949, 0.482, 0.216}
\definecolor{yafcolor3}{rgb}{0.47, 0.549, 0.306}
\definecolor{yafcolor4}{rgb}{0.925, 0.165, 0.224}
\definecolor{yafcolor5}{rgb}{0.141, 0.345, 0.643}
\definecolor{yafcolor6}{rgb}{0.965, 0.933, 0.267}
\definecolor{yafcolor7}{rgb}{0.627, 0.118, 0.165}
\definecolor{yafcolor8}{rgb}{0.878, 0.475, 0.686}
\node[exnode] (n7) at (1.6143, 6.6141) {tequila};
\node[exnode] (n3) at (3.3474, 7.8522) {\em pineapple\\juice};
\node[exnode] (n9) at (4.2298, 4.1479) {triple\\sec};
\node[exnode] (n90) at (3.6375, 10.336) {coffee\\liqueur};
\node[exnode] (n20) at (4.6667, 9.3211) {rum};
\node[exnode] (n310) at (4.375, 1.4) {lillet};
\node[exnode] (n200) at (1.8146, 8.681) {coconut\\liqueur};
\node[exnode] (n1) at (3.0456, 6.0531) {\em orange\\juice};
\node[exnode] (n42) at (2.461, 10.056) {dark rum};
\node[exnode] (n5) at (2.375, 1.40578) {gin};
\node[exnode] (n110) at (2.0531, 4.40405) {apple\\juice};
\node[exnode] (n0) at (2.3225, 2.7986) {\emph{vodka}};
\node[exnode] (n32) at (4.9609, 7.6152) {lemon};
\node[exnode] (n4) at (4.6134, 5.4846) {grenadine};

\draw[exedge] (n0) edge (n1);
\draw[exedge] (n0) edge (n310);
\draw[exedge] (n0) edge (n5);
\draw[exedge] (n0) edge (n110);
\draw[exedge] (n1) edge (n9);
\draw[exedge] (n1) edge (n3);
\draw[exedge] (n1) edge (n4);
\draw[exedge] (n1) edge (n7);
\draw[exedge] (n3) edge (n32);
\draw[exedge] (n3) edge (n200);
\draw[exedge] (n3) edge (n42);
\draw[exedge] (n3) edge (n20);
\draw[exedge] (n3) edge (n90);


\draw[fitline, yafcolor1]
	(n0.north east) --
	(n1.north east) -- (n1.north west) -- (n1.south west) --
	(n110.south west) -- (n110.south east) --
	(n0.south east) -- cycle;

\draw[fitline, yafcolor2]
	(n0.south east) --
	(n5.south west) -- (n5.south east) --
	(n310.south east) -- (n310.north east) -- (n310.north west) --
	(n0.north west) -- (n0.south west) -- cycle;

\draw[fitline, yafcolor4]
	(n4.south east) -- (n4.north east) -- (n4.north west) --
	(n7.north west) -- (n7.south west) -- (n7.south east) --
	cycle;

\draw[fitline, inner sep = 2pt, yafcolor3]
	(n4.north east) -- (n4.north west) --
	(n3.north west) -- (n3.south west) --
	(n7.south west) -- (n7.south east) --
	(n9.south east) -- (n9.north east) --
	cycle;

\draw[fitline, yafcolor5]
	(n32.south east) -- (n32.north east) -- (n32.north west) --
	(n20.north west) --
	(n90.north west) -- (n90.south west) --
	(n42.south west) --
	(n200.south west) -- (n200.south east) -- 
	(n3.south east) -- (n3.north east) -- cycle;

\node[labnode, yafcolor5] (l5) at (0, 11.1)  {};
\node[inner sep = 1pt, font = \scriptsize, right = 1pt of l5.east, anchor = mid west] (t5) {bahama mama};
\node[labnode, yafcolor3, right = 1pt of t5.mid east, anchor = west] (l3) {};
\node[inner sep = 1pt, font = \scriptsize, right = 1pt of l3.east, anchor = mid west] (t3) {harvest moon};
\node[labnode, yafcolor4, right = 1pt of t3.mid east, anchor = west] (l4) {};
\node[inner sep = 1pt, font = \scriptsize, right = 1pt of l4.east, anchor = mid west] (t4) {tequila sunrise};
\node[labnode, yafcolor1, below = 2pt of l5.south, anchor = north] (l1) {};
\node[inner sep = 1pt, font = \scriptsize, right = 1pt of l1.east, anchor = mid west] (t1) {screwdriver};
\node[labnode, yafcolor2, right = 1pt of t1.mid east, anchor = west] (l2) {};
\node[inner sep = 1pt, font = \scriptsize, right = 1pt of l2.east, anchor = mid west] (t2) {vesper};

\end{tikzpicture}

\caption{\label{figure:snippet}A snippet of the discovered network for the \cocktails\ dataset.}
\end{figure}
\begin{figure}[t]
\begin{center}
\begin{tikzpicture}[baseline]
\begin{axis}[xlabel={base degree}, ylabel= {reconstructed degree},
    height = 5cm,
    width = \columnwidth,
    cycle list name=yaf,
    xmin = 0,
	legend entries = {\starprb, \densprb},
	legend pos = north west
    ]


\addplot[yafcolor2, only marks, mark=*, mark size = 0.8pt,
	error bars/.cd, y dir = plus, y explicit, error bar style = {densely dotted, black}]
	table[x index = 0, y index = 1, y error expr = {\thisrowno{2} - \thisrowno{1}}, header = false]  {cocktails_degrees.dat};
\addplot[yafcolor1, only marks, mark=x, mark size = 2pt, mark options = {semithick, opacity = 0.8},] table[x index = 0, y index = 2, header = false]  {cocktails_degrees.dat};

\pgfplotsextra{\yafdrawaxis{0}{200}{0}{200}}
\end{axis}
\end{tikzpicture}
\end{center}
\caption{\label{figure:cocktail-degrees}Vertex degree of reconstructed networks as a function of the vertex
degree in the base network in the {\cocktails} dataset. 
}
\vspace{-2mm}
\end{figure}
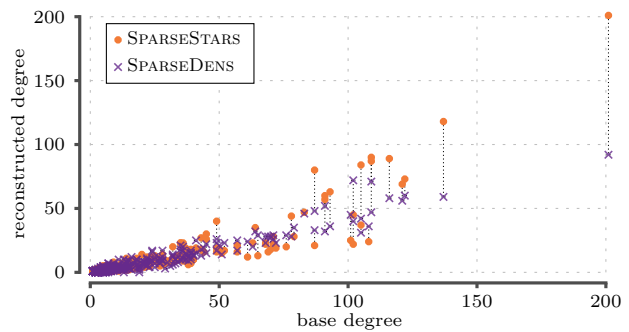

\section{Related work}
\label{sec:related}

To the best of our knowledge, we are the first to introduce and study
the {\densprb} and {\starprb} problems.  
As a result, the problem definitions,
technical results, and algorithms presented in this paper are novel.
However, our problems are clearly related to 
\emph{network sparsification} and \emph{network design} problems. 

\spara{Network sparsification:} 
Existing work on network sparsification aims to simplify the input network 
--- by removing edges --- such that
the remaining network maintains some of the properties of the original network. 
Such properties include
shortest paths and connectivity~\cite{elkin2005approximating,zhou10network},
cuts~\cite{ahn12graph,fung11general},
source-to-sink flow~\cite{misiolek06two}, 
spectral properties~\cite{ahn13spectral,batson13spectral,spielman11graph}, 
modularity~\cite{arenas2007size}, as well as
information-propagation pathways~\cite{mathioudakis11sparsification}.
Other work focuses on sparsification that improves network visualization~\cite{lindner2015structure}.
The main difference of our paper and existing work is that 
we consider sparsifying the network while maintaining the structure of a given set of communities. 
Such community-aware sparsification raises new computational challenges that are distinct from the computational
problems studied in the past. 

\spara{Network design problems:} 
At a high level, 
network-design problems consider a set of constraints 
and ask to construct a minimum-cost network that satisfies those constraints~\cite{alon06general,angluin13connectivity,gupta12online,korach08complete,korach03clustering}. 
As in our case, cost is usually formulated in terms of the number of edges, 
or total edge weight.
Many different constraints have been considered in the literature: 
reachability, 
connectivity, 
cuts, flows, etc.
%
Among the existing work in network design, 
the most related to our paper is the work by Angluin~{\etal}~\cite{angluin13connectivity}
and by Korach and Stern~\cite{korach03clustering,korach08complete}. 
Using the notation introduced in Section~\ref{sec:prel}, 
Angluin~{\etal}\ essentially solve the {\setnetworkprb} problem 
with the~$\indicator_{\connectivity}$ property. 
Our results on the {\densprb} problem and its variant on connected {\densprb}
are largely inspired by the work of Angluin et al.
On the other hand, 
for the {\starprb} problem we need to introduce completely new techniques, 
as the submodularity property is not applicable in this case. 
Korach and Stern~\cite{korach03clustering,korach08complete} study the following problem: 
given a collection of sets, construct a minimum-weight {\em tree}
so that each given set defines a star in this tree.  
Clearly, this problem is related to the {\starprb} problem considered here, 
however, the tree requirement create a very significant differentiation: 
the problem studied by Korach and Stern is polynomially-time solvable, 
while~{\starprb} is \NP-hard.
In terms of real-world applications, 
while tree structures are well motivated in certain cases (e.g., overlay networks), 
they are not natural in many other (e.g., social networks).

\section{Concluding remarks}
\label{sec:conclusions}

In this paper, we have introduced {\setnetworkprb}, 
a new formulation of network sparsifcation, 
where the input consists not only of a network
but also of a set of communities.
The goal in {\setnetworkprb} is twofold: 
($i$) sparsify the input network as much as possible, 
and ($ii$) guarantee some connectedness property for the subgraphs induced by the input communities
on the sparsifiers. 
We studied two connectedness properties 
and showed that 
the corresponding instances of {\setnetworkprb} is {\NP}-hard. 
We then designed effective approximation algorithms for both problems.
Our experiments with real datasets obtained from diverse domains, 
verified the effectiveness of the proposed algorithms, in terms of the number of edges they removed. They also demonstrated that the obtained sparsified networks provide interesting insights about the
structure of the original network with respect to the input communities.
%

\smallskip
\spara{Acknowledgements.}
This work was funded by 
Tekes project ``Re:Know,'' 
Academy of Finland project ``Nestor'' (286211),
EC H2020 RIA project ``SoBigData'' (654024), 
a Microsoft gift, a Nokia Faculty Fellowship,
and by NSF grants: IIS 1320542, IIS 1421759 and CAREER~1253393.

{
\bibliography{thebibfile-terse}

\begin{thebibliography}{10}

\bibitem{ahn12graph}
K.~J. Ahn, S.~Guha, and A.~McGregor.
\newblock Graph sketches: sparsification, spanners, and subgraphs.
\newblock In {\em {PODS}}, 2012.

\bibitem{ahn13spectral}
K.~J. Ahn, S.~Guha, and A.~McGregor.
\newblock Spectral sparsification in dynamic graph streams.
\newblock In {\em {RANDOM-APPROX}}, pages 1--10, 2013.

\bibitem{alon06general}
N.~Alon, B.~Awerbuch, Y.~Azar, N.~Buchbinder, and J.~Naor.
\newblock A general approach to online network optimization problems.
\newblock {\em {ACM} Transactions on Algorithms}, 2(4):640--660, 2006.

\bibitem{angluin13connectivity}
D.~Angluin, J.~Aspnes, and L.~Reyzin.
\newblock Network construction with subgraph connectivity constraints.
\newblock {\em J. of Comb. Optimization}, 2013.

\bibitem{arenas2007size}
A.~Arenas, J.~Duch, A.~Fern{\'a}ndez, and S.~G{\'o}mez.
\newblock Size reduction of complex networks preserving modularity.
\newblock {\em New Journal of Physics}, 9(6):176, 2007.

\bibitem{batson13spectral}
J.~D. Batson, D.~A. Spielman, N.~Srivastava, and S.~Teng.
\newblock Spectral sparsification of graphs: theory and algorithms.
\newblock {\em CACM}, 56(8):87--94, 2013.

\bibitem{elkin2005approximating}
M.~Elkin and D.~Peleg.
\newblock Approximating $k$-spanner problems for $k>2$.
\newblock {\em Theoretical Computer Science}, 337(1):249--277, 2005.

\bibitem{Farine150057}
D.~Farine.
\newblock The role of social and ecological processes in structuring animal
  populations.
\newblock {\em Royal Society Open Science}, 2(4), 2015.

\bibitem{fung11general}
W.~S. Fung, R.~Hariharan, N.~J. Harvey, and D.~Panigrahi.
\newblock A general framework for graph sparsification.
\newblock In {\em STOC}, 2011.

\bibitem{garey79computers}
M.~Garey and D.~Johnson.
\newblock {\em {Computers and intractability: a guide to the theory of
  NP-completeness}}.
\newblock WH Freeman \& Co., 1979.

\bibitem{gupta12online}
A.~Gupta, R.~Krishnaswamy, and R.~Ravi.
\newblock Online and stochastic survivable network design.
\newblock {\em {SIAM} Journal of Computing}, 41(6):1649--1672, 2012.

\bibitem{korach03clustering}
E.~Korach and M.~Stern.
\newblock The clustering matroid and the optimal clustering tree.
\newblock {\em Math. Program.}, 98(1-3):385--414, 2003.

\bibitem{korach08complete}
E.~Korach and M.~Stern.
\newblock The complete optimal stars-clustering-tree problem.
\newblock {\em Discrete Applied Mathematics}, 156(4):444--450, 2008.

\bibitem{lindner2015structure}
G.~Lindner, C.~L. Staudt, M.~Hamann, H.~Meyerhenke, and D.~Wagner.
\newblock Structure-preserving sparsification of social networks.
\newblock In {\em ASONAM}, 2015.

\bibitem{mathioudakis11sparsification}
M.~Mathioudakis, F.~Bonchi, C.~Castillo, A.~Gionis, and A.~Ukkonen.
\newblock Sparsification of influence networks.
\newblock In {\em KDD}, 2011.

\bibitem{mestre2006greedy}
J.~Mestre.
\newblock Greedy in approximation algorithms.
\newblock In {\em ESA}, 2006.

\bibitem{misiolek06two}
E.~Misiolek and D.~Z. Chen.
\newblock Two flow network simplification algorithms.
\newblock {\em IPL}, 97(5):197--202, 2006.

\bibitem{satuluri2011local}
V.~Satuluri, S.~Parthasarathy, and Y.~Ruan.
\newblock Local graph sparsification for scalable clustering.
\newblock In {\em SIGMOD}, 2011.

\bibitem{spielman11graph}
D.~A. Spielman and N.~Srivastava.
\newblock Graph sparsification by effective resistances.
\newblock {\em {SIAM} J. Comput.}, 40(6):1913--1926, 2011.

\bibitem{Wolsey:1982jx}
L.~Wolsey.
\newblock {An analysis of the greedy algorithm for the submodular set covering
  problem.}
\newblock {\em Combinatorica}, 2(4):385--393, 1982.

\bibitem{zhou10network}
F.~Zhou, S.~Mahler, and H.~Toivonen.
\newblock Network simplification with minimal loss of connectivity.
\newblock In {\em ICDM}, 2010.

\end{thebibliography}
}

\clearpage 


\section*{\LARGE Supplementary material}
\bigskip

\section*{\Large Community-aware \\ network sparsification}
\bigskip

\section*{Aristides Gionis \\ 
Polina Rozenshtein \\
Nikolaj Tatti \\
Evimaria Terzi }
\bigskip
\bigskip
\bigskip

\section*{S1.$\quad$Proofs} 

\begin{proof}[Proof of Proposition~\ref{proposition:density-nphardness}]
We consider the decision version of \densprb.
The problem is obviously in \NP.  
To prove the hardness we provide a reduction from from the \hitprb\ problem. 
An instance of \hitprb\ consists of a universe of items $U$, 
a collection of sets $C_i\subseteq U$, 
and a positive integer $c$.  
The task is to decide whether there exists a ``hitting set'' $X\subseteq U$, 
of cardinality at most $c$,  
such that $C_i\cap X\neq \emptyset$, for every $i$.

Consider an instance of the {\hitprb} problem. 
We will show how to obtain a solution for this instance, using \densprb.
We proceed as follows. 
First we create a graph $G_0=(V,E_0)$, 
such that $|V|=|U| + 1$: 
for every item $u\in U$ we create a vertex $x_u\in V$,  
in addition we add one special vertex $s\in V$.  
The graph is fully connected, $\abs{E_0} = {\abs{V} \choose 2}$.

Now for every set $C_i$ in the instance of {\hitprb} we create
a set of vertices $S_i$ for our problem such that $S_i=\set{s} \cup \set{x_u \mid u\in C_i}$. 
We also set $\alpha_i = \alpha = {\abs{V} \choose 2}^{-1}$. Note that $\alpha$ is so low
that to satisfy the constraint $\indicator_{\density\ge\alpha}(G,S_i)$
it is sufficient to have $E(S_i) \geq 1$.

Let $G= (V, E)$ be a solution for {\densprb}, if one exists.
We can safely assume that each edge in $E$ is adjacent to $s$.
To see this, assume that $e = (x_u, x_v) \in E$ and modify 
$E$ by adding $(x_u, s)$, if not already in $E$, and deleting $e$.
By doing this swap, we do not violate any constraint since any $S_i$ that contains $e$ will also contain $(x_u, s)$
and having one edge is enough to satisfy the constraint. 
Moreover, we do not increase the number of edges in $E$.

Using this construction, we can see that the adjacent vertices in $E$, excluding $s$, correspond to the hitting set;
that is,
there exists a solution to the {\hitprb} problem of
cardinality at most $c$
if and only if there exists a solution to {\densprb} that uses at most $c$ edges.
\qed
\end{proof}

It is interesting to observe that
our proof implies that the \densprb\ problem is \NP-hard 
even if the feasibility network $G_0$ is fully-connected.


\begin{proof}[Proof of Proposition~\ref{prop:submodularity}]
Showing that \densitypotential\ is monotone is quite straightforward, 
so we focus on submodularity. 
We need to show that 
for any set of edges $X\subseteq Y\subseteq E_0$ 
and any edge $e \notin Y$ it is
\[
\densitypotential(Y\cup\{e\})-\densitypotential(Y) \leq \densitypotential(X\cup\{e\})-\densitypotential(X).
\]
Since \densitypotential\ is a summation of terms, 
as per Equation~\eqref{eq:densitypotential}, 
it is sufficient to show that each individual term is submodular. 
Thus, we need to show that for any $S_i\in\sets$, 
$X\subseteq Y\subseteq E_0$, and $e\notin Y$ it is
\[
\densitypotential(Y\cup\{e\},S_i)-\densitypotential(Y,S_i) \leq 
\densitypotential(X\cup\{e\},S_i)-\densitypotential(X,S_i).
\]
To show the latter inequality, 
first observe that for any 
$S_i, Z\subseteq E_0$ and $e\notin Z$ the difference 
$\densitypotential(Z\cup\{e\},S_i)-\densitypotential(Z,S_i)$ 
is either $0$ or $1$. 
Now fix $S_i$, $X\subseteq Y\subseteq E_0$, and $e\notin Y$;
if $\densitypotential(X\cup\{e\},S_i)-\densitypotential(X,S_i)=0$, 
either the set of edges $X$ satisfy the density constraint on $S_i$,
or $e$ is not incident in a pair of vertices in $S_i$. 
In the latter case, 
$\densitypotential(Y\cup\{e\},S_i)-\densitypotential(Y,S_i)=0$, as well.
In the former case, if $X$ satisfies the density constraint, 
since $X\subseteq Y$, then 
the set of edges $Y$ should also satisfy the density constraint, and
thus 
$\densitypotential(Y\cup\{e\},S_i)-\densitypotential(Y,S_i)=0$. 
\qed
\end{proof}


\begin{proof}[Proof of Proposition~\ref{prop:starhardness}]
We consider the decision version of the \starprb\ problem.
Clearly the problem is in \NP. 
To prove the completeness we will obtain a reduction from the {\matchprb} problem, 
the 3-dimensional complete matching problem~\cite{garey79computers}. 
An instance of {\matchprb}
consists of three disjoint finite sets $X$, $Y$, and $Z$, having the same size, 
and a collection of $m$ sets $\mathcal{C} = \{C_1, \ldots, C_m\}$ 
containing exactly one item from $X$, $Y$, and $Z$, so that $\abs{C_i} = 3$. 
The goal is to decide whether there exists a subset of $\mathcal{C}$ where each set
is disjoint and all elements in $X$, $Y$, and $Z$ are covered.

Assume an instance of {\matchprb}. 
For each $C_i$ create four vertices $p_i$, $u_i$, $v_i$, and $w_i$. 
Set the network $G_0=(V, E_0)$ to be a fully connected graph over all those vertices. 
Define $P = \set{p_i}$, the set of $p_i$'s.
For each $x \in X$, create a set $S_x = \set{p_i, u_i, v_i \mid x \in C_i}$.
Similarly, for each $y \in Y$, create a set $S_y = \set{p_i, u_i, w_i \mid y \in C_i}$
and, for each $z \in Z$, create a set $S_z = \set{p_i, v_i, w_i \mid z \in C_i}$.
Let $\sets$ consist of all these sets.

Let $G=(V,E)$ be the optimal solution to the {\starprb} problem; 
such a solution will consist of induced subgraphs $G_i=\left(S_i,E(S_i)\right)$ that contain a star.
Let $\mu$ be the function mapping each $S_i$ to a vertex that acts as a center of the star defined by $G_i$.
Let
$O = \set{\mu(S_i); S_i \in \sets}$ be the set of these center vertices in the optimal solution.
We can safely assume that  $O \subseteq P$; 
even if in the optimal solution there exists an $S_i$ with $E(S_i)$ not intersecting with any 
other $E(S_j)$, $p_i$ can be picked as the center of
this star. For each $o \in O$, define $\mathcal{N}_o = \set{S_i \in \sets \mid \mu(S_i) = o}$.

The number of edges $\abs{E}$ in the optimal graph $G=(V,E)$
is equal to $\sum_{S_i \in \sets} (\abs{S_i} - 1)-D$, 
where $D$ is the number the edges that are counted in more than one star. 
To account for this double counting we proceed as follows: 
if $\mathcal{N}_o$ contains two sets, then 
there is one edge adjacent to $o$ that is counted twice. 
If $\mathcal{N}_o$ contains three sets, then
there are three edges adjacent to $o$ that are counted twice. 
This leads to
\[
\abs{E} = \sum_{S_i \in \sets} (\abs{S_i} - 1) - 
\sum_{o \in O} \left(\mathbb{I}_{[\abs{\mathcal{N}_o} = 2]} + 3\mathbb{I}_{[\abs{\mathcal{N}_o} = 3]} \right),
\]
where $\mathbb{I}$ is the indicator function with $\mathbb{I}_{[A]} = 1$ if the statement $A$ is true, 
and $0$ otherwise.

To express the number of edges solely with a sum over the sets,
we define a function $f$ as $f(1) = 0$, $f(2) = 1/2$ and $f(3) = 1$. Then
\[
\abs{E} =  \sum_{S_i \in \sets} \left(\abs{S_i} - 1  - f(\abs{\mathcal{N}_{\mu(S_i)}})\right).
\]

Let $\mathcal{Q} \subseteq \mathcal{C}$ be the set $3$-dimensional edges 
corresponding to the set of selected star centers $O$.
Set $t = \sum_{S_i \in \sets} \left(\abs{S} - 2\right)$. 
Then $\abs{E} \leq t$ if and only if
every $\mathcal{N}_o$ contains 3 sets,  
which is equivalent to $\mathcal{Q}$ containing disjoint sets that cover $X$ and $Y$ and $Z$.
\qed
\end{proof}


\begin{proof}[Proof of Observation~\ref{observation:relationship}]
The first part of the inequality
follows from the fact that any solution 
$(V,D)$ for {\distarprb}
can be translated to a solution for {\starprb} by simply ignoring the edge directions
and removing duplicates, if needed. 
The second part of the inequality follows from the fact that any solution $(V,E)$ for {\starprb}
can be translated to a feasible solution for {\distarprb}, 
with at most two times as many edges, 
by creating two copies each edge $(x,y)$ in $E$: 
one for $(x \rightarrow y)$ and one for $(y \rightarrow x)$. 
\qed
\end{proof}



To prove Proposition~\ref{prop:optimal} we will use the following lemma which we state without the proof.

\begin{lemma}
\label{lem:star}
Let $H \in \mathcal{H}$ be a hyper-edge and let $T$ be a star
with the center $x$ in $\intersection(H)$ connecting to every vertex
in $\union(H) \setminus \set{x}$. The number of edges in $T$ is equal to
\[
	\sum_{C \in H} (\abs{C} - 1) - c(H).
\]
\end{lemma}

Now we are ready to prove Proposition~\ref{prop:optimal}.

\begin{proof}[Proof of Proposition~\ref{prop:optimal}]
Let us first prove $\abs{D^\ast} \leq \abs{D_\mathcal{I}}$.
By definition,
$\mathcal{H}$ contains all sets $C_i$ as singleton groups.
Therefore, each set $C_i$ is included in some $H \in \mathcal{I}$.
Hence, $D_\mathcal{I}$ is a feasible solution for {\distarprb} and
therefore $\abs{D^\ast} \leq \abs{D_\mathcal{I}}$.

We will now prove the other direction.
By definition, $D^\ast$ is a union of stars $\set{T_i}$,
where each $T_i = (C_i, A_i)$.  Define a family $\mathcal{P}$ by grouping each
$C_i$ sharing the same star center.  Note that $\mathcal{P}$ is a disjoint subset
of $\mathcal{H}$, consequently, it is a feasible solution for \hyperprb.
Lemma~\ref{lem:star} implies that
\[
\begin{split}
\abs{D^\ast} & = \sum_{H \in \mathcal{P}} \sum_{C \in H} \abs{C} - 1 - c(H) \\
& = \sum_{C \in \mathcal{C}} (\abs{C} - 1) - \sum_{H \in \mathcal{P}} c(H) \\
& \geq \sum_{C \in \mathcal{C}} (\abs{C} - 1) - \sum_{H \in \mathcal{I}} c(H), \\
\end{split}
\]
where the first equality follows from the fact that the joined trees are edge-disjoint.

Lemma~\ref{lem:star} implies that
\[
\abs{D_{\mathcal{I}}}
\leq \sum_{H \in \mathcal{I}} \sum_{C \in H} \abs{C} - 1 - c(H)
= \sum_{C \in \mathcal{C}} (\abs{C} - 1) - \sum_{H \in \mathcal{I}} c(H),
\]
where the last equality follows since
each set $C_i$ is included in some $H \in \mathcal{I}$.
\qed
\end{proof}

\begin{proof}[Proof of Proposition~\ref{prop:kapprox}]
The set of feasible solutions of the {\hyperprb} problem
forms a $k$-extensible system~\cite{mestre2006greedy}. 
As shown by Mestre~\cite{mestre2006greedy}, 
the greedy algorithm provides a $k$-factor approximation
to the problem of finding a solution with the maximum weight in a $k$-extensible system.
\qed
\end{proof}

\begin{proof}[Proof of Proposition~\ref{prop:undirbound}]
For the solution of the {\matching} problem
we know that $\abs{E}\leq |D_\mathcal{J}|$.
This is because we can obtain $E$ from $D_\mathcal{J}$ 
by ignoring edge directions, and possibly removing edges, if needed. 
From the latter inequality, 
Observation~\ref{observation:relationship}, and Corollary~\ref{cor:dirbound},
the statement follows.
\qed
\end{proof}

\section*{S2.$\quad$Extension to weighted networks}

Our problem definition can be extended for \emph{weighted} graphs
$G=(V,E,d)$, 
where $V$ and $E$ are the sets of nodes and edges in the network. 
In this case, 
we assume that edges are weighted by a distance function $d:E\rightarrow \Reals_{+}$.
Small distances indicate strong connections
and large distances indicate weak connections. 
The distance of an edge $e\in E$ is denoted by $d(e)$, 
while the total distance of a set of edges $E'\subseteq E$
is defined as $d(E')=\sum_{e\in E'}d(e)$.
%
%
Given such a weighted network, 
we can extend the definition of the {\setnetworkprb} problem
as follows:
\begin{problem}[{\weightedproblem}]
\label{problem:weighted}
Consider an underlying network $G=(V, E, d)$ , 
where $d$ represent edge distances, 
and let $\Property$ be a graph property.
Given a set of $\ell$ communities
$\communities = \{\community_1,\ldots, \community_\ell\}$, 
we want to construct a {\em sparse network} $G' = (V, E')$, such that, 
($i$) $E' \subseteq E$; 
($ii$) $\indicator_\Property(G',\community_i) = 1$, for all $\community_i\in\communities$;  
and 
($iii$) the sum of distances of edges in the sparse network, 
$d(E')=\sum_{e\in E'}d(e)$, is minimized. 
\end{problem}

As before, 
depending on whether $\Property$ is the connectivity, 
the density or the star-containment property,
we get the corresponding weighted versions of the 
{\connectedprb}, {\densprb} and {\starprb} problems respectively.
The greedy algorithms developed for the {\connectedprb} and {\densprb} problems can be also used
for their weighted counterparts.
In particular, in the greedy step of the algorithm the next edge is chosen
so as to maximize the potential difference
\[
\frac{\densitypotential(E'\cup\{e\}) - \densitypotential(E')}{d(e)}. 
\]
However, the algorithm we give for {\starprb}
is only applicable to unweighted networks; 
developing a new algorithm for the weighted case is left as future work.

\section*{S3.$\quad$Birds case study}

We present a second case study 
where the input communities are group sightings of birds.
We run the {\matching} algorithm on the input sets, 
and we obtain a sparsified network with $809$ star centers and $21\,077$ edges, 
that is, around $47$\,\% of the edges of input network.  
The dataset also contains gender (male/female/unknown), 
age (juvenile/adult), and immigration status of each individual bird. 
We studied whether
some characteristics are favoured when selecting centers. Here, we found out
that juveniles are preferred as centers, as well as, male residents, 
see Figure~\ref{figure:birds-stats}.

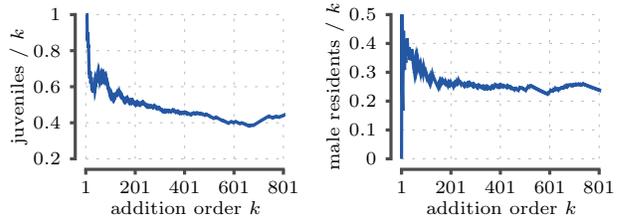
\begin{figure}[b]
\begin{center}
\begin{tikzpicture}[baseline]
\begin{axis}[xlabel={addition order $k$}, ylabel= {juveniles / $k$},
    height = 3.5cm,
    width = 0.5\columnwidth,
    cycle list name=yaf,
    xmin = 1,
    ymin = 0.2,
    ymax = 1,
	xtick = {1, 201, 401, 601, 801}
    ]

\addplot[yafcolor5, no markers]
	table[x expr = {\coordindex + 1}, y index = 0, header = false]  {birds_stats.dat};

\pgfplotsextra{\yafdrawaxis{1}{800}{0.2}{1}}
\end{axis}
\end{tikzpicture}
\begin{tikzpicture}[baseline]
\begin{axis}[xlabel={addition order $k$}, ylabel= {male residents / $k$},
    height = 3.5cm,
    width = 0.5\columnwidth,
    cycle list name=yaf,
    xmin = 1,
    ymin = 0,
    ymax = 0.5,
	ytick = {0, 0.1, 0.2 , 0.3, 0.4, 0.5},
	xtick = {1, 201, 401, 601, 801}
    ]

\addplot[yafcolor5, no markers]
	table[x expr = {\coordindex + 1}, y index = 1, header = false]  {birds_stats.dat};

\pgfplotsextra{\yafdrawaxis{1}{800}{0}{0.5}}
\end{axis}
\end{tikzpicture}
\end{center}
\caption{\label{figure:birds-stats}
Proportion of juveniles and male residents in top-$k$ selected star centers in \birds as a function of $k$. }
\end{figure}

\end{document}